\documentclass[journal]{IEEEtran}

\usepackage{changes}

\usepackage{amsmath,amsfonts}

\usepackage{amsthm}  
\newtheorem{theorem}{Proposition}

\usepackage{algpseudocode}
\usepackage{algorithm}
\usepackage{array}
\usepackage{multirow}
\usepackage{booktabs} 

\RequirePackage{fix-cm} 
\usepackage{hyperref}
\usepackage{textcomp}
\usepackage{stfloats}
\usepackage{url}
\usepackage{verbatim}
\usepackage[font=footnotesize]{caption}
\usepackage{graphicx}
\usepackage{pgfplots}
\pgfplotsset{compat=1.18}
\usepackage{subcaption}  
\usepackage{tikz}
\usetikzlibrary{positioning, shapes, decorations,calc}
\usepackage{cite}
\usepackage{eurosym} 

\usepackage{xcolor}

\begin{document}

\title{Decision-Focused Learning for Neural Network-Constrained HVAC Scheduling}

\author{Pietro~Favaro,~\IEEEmembership{Student Member,~IEEE,}
        Jean-François~Toubeau,~\IEEEmembership{Member,~IEEE,}
        François~Vallée,~\IEEEmembership{Member,~IEEE,}
        and~Yury~Dvorkin,~\IEEEmembership{Member,~IEEE,}
\thanks{P. Favaro, J.F. Toubeau, F. Vallée are with the University of Mons, Belgium}
\thanks{Y. Dvorkin is with Johns Hopkins University, Baltimore, USA}
}

\markboth{Accepted in IEEE Trans. Smart Grid}%
{Shell \MakeLowercase{\textit{et al.}}: Bare Demo of IEEEtran.cls for IEEE Journals}


\maketitle

\begin{abstract}
Heating, Ventilation, and Air Conditioning (HVAC) is a major electricity end-use with a substantial potential for providing grid services, such as demand response. {Harnessing this flexibility requires accurate modeling of the thermal dynamics of buildings, a difficult task because nonlinear heat transfer and recurring daily cycles make historical data highly correlated and insufficient to generalize to new weather, occupancy, and control scenarios. This paper presents an HVAC management system formulated as a Mixed Integer Quadratic Program (MIQP), where Neural Network (NN) models of thermal dynamics are embedded as exact mixed-integer linear constraints. Unlike traditional training approaches that minimize prediction errors, we employ Decision-Focused Learning (DFL) to learn the NN parameters with the objective of directly improving the HVAC cost performance. However, the discrete nature of MIQP hinders DFL, as it leads to undefined and discontinuous gradients, thus impeding standard gradient-based training.}{} We leverage Stochastic Smoothing (SS) to enable efficient gradient computation without the need to differentiate the MIQP. Experiments on a realistic five-zone building using a high-fidelity simulator demonstrate that the proposed SS-DFL approach outperforms conventional {identify-then-optimize (i.e., the thermal dynamics model is identified on historical data then used in optimization)}{} and relaxed DFL methods in both cost savings and grid service performance, highlighting its potential for scalable, grid-aware building control.
\end{abstract}


\IEEEpeerreviewmaketitle

\section{Introduction}
\label{sec: intro}

\IEEEPARstart{B}{uildings}
account for over 30\% of global energy use \cite{noauthor_buildings_2023}, making them central to the low-carbon energy transition. Among their energy-intensive components, Heating, Ventilation, and Air Conditioning (HVAC) accounts for around 65\% of total energy consumption in European households \cite{noauthor_energy_2024}, yet it remains a largely untapped source of flexibility \cite{papadaskalopoulos_decentralized_2013, kouzelis_estimation_2015, mocanu_energy_2016}. Controlling HVAC settings can reshape energy consumption patterns \cite{tian_real-time_2021}, improve efficiency \cite{m_yousefi_predictive_2021}, and even support grid stability \cite{kim_experimental_2016, wang_online_2012}. However, unlocking this potential requires understanding and controlling building thermal dynamics, i.e., an intricate problem where physics, data, and optimization must be seamlessly integrated into the decision-making pipeline.

A major opportunity to leverage HVAC flexibility is the participation in the day-ahead electricity market, the main platform for trading electricity between grid actors. In Europe, this market is cleared daily, setting electricity volumes and prices for each hour of the following day \cite{noauthor_current_2015}. This timeline aligns naturally with the slow thermal dynamics of buildings, where temperature adjustments from HVAC decisions and external conditions unfold gradually over time. 
Capturing these effects is crucial to maintain occupants' comfort \cite{m_yousefi_predictive_2021}. Thermal models for buildings generally fall into two main categories: physics-based and data-driven\footnote{
    A comprehensive review of the literature is provided in \cite{drgona_all_2020}, Section~III.
    }.

Physics-based (i.e., white-box) models rely on detailed information about a building’s physical characteristics, including construction materials, insulation properties, and ventilation dynamics. This category includes high-fidelity simulators such as EnergyPlus\cite{crawley_energyplus_2001}.
While these models provide highly detailed representations of buildings \cite{afroz_modeling_2017}, some parameters, such as the materials resistance and capacitance, are inherently uncertain due to factors like aging and variability in the installation process, which can lead to significant errors \cite{dagostino_experimental_2022}. Moreover, it requires expertise in building thermal modeling, which may not be readily available at each building site. But, there also exist simpler models: aggregated (or lumped) Resistance-Capacitance (RC) models. These models are linear and based on circuit analogy between heat transfer and electricity. The parameters of these models can be computed based on the construction materials \cite{kircher_lumped_2015}, or by data-driven regression \cite{f_belic_thermal_2016}.

In contrast, data-driven (i.e., black-box) models rely on statistical or Machine Learning (ML) techniques trained on historical weather and energy consumption data. These models range from simple linear regressions \cite{qiang_improved_2015} to highly complex nonlinear architectures, such as Neural Network (NN) models \cite{y_-j_kim_supervised-learning-based_2020, s_a_nabavi_deep_2021}. In~\cite{drgona_physics-constrained_2021}, physics-informed NNs were employed for control-oriented predictions. However, these NN approaches did not consider day-ahead scheduling, which introduces fundamentally different requirements. \autoref{tab: thermal models} summarizes the thermal modeling approaches.

\newcolumntype{L}[1]{>{\raggedright\arraybackslash}p{#1}}
\begin{table}[thb]
    \centering
    \caption{Comparison of thermal dynamics modeling approaches.}
    \label{tab: thermal models}
    \begin{tabular}{p{1em} L{5.5em} L{4.5em} L{8em} L{4.2em}}
    \toprule
     & \textbf{Model} & \textbf{Complexity} & \textbf{Data Requirement} & \textbf{Ref.} \\
    \midrule
    \multirow{4}{*}{\rotatebox[origin=c]{90}{White box}} 
      & High-fidelity simulator 
      & Very high 
      & Geometry, materials, weather 
      & \cite{crawley_energyplus_2001, afroz_modeling_2017, dagostino_experimental_2022} \\
      & Lumped RC (known R/C)
      & Low 
      & R and C 
      & \cite{kircher_lumped_2015} \\ \hline
    \multirow{5}{*}{\rotatebox[origin=c]{90}{Gray box}} & Lumped RC (parameter fitting) & Moderate 
      & I/O measurements for R/C &  \cite{f_belic_thermal_2016} \\
       & Physics-informed NN & High & Partial physics + data & \cite{drgona_physics-constrained_2021} \\\hline
    \multirow{4}{*}{\rotatebox[origin=c]{90}{Black box}} & NN & Moderate & Historical input–output data & \cite{y_-j_kim_supervised-learning-based_2020, s_a_nabavi_deep_2021} \\
     & Regression models & Low 
      & Historical system data & \cite{qiang_improved_2015} \\
    \bottomrule
    \end{tabular}
\end{table}

NNs can capture complex patterns with a limited manual modeling effort, and will thus be used in this paper. To fully exploit the capabilities of NN models for HVAC control, it is essential to unify modeling and control tasks within a single framework. Such an integrated approach would allow HVAC systems to make energy-efficient planning decisions while maintaining occupants' comfort.
{
A first promising avenue is to model the control policy by a NN that outputs control decisions, allowing it to internalize the building dynamics. Traditionally, such NNs in real-time control are trained using reinforcement learning (RL) \cite{jang_active_2024}, which directly optimizes decisions through interaction with the environment. Alternatively, self-supervised learning can be employed, but, similarly to RL, it often lacks guarantees on solution quality, especially in the presence of hard physical and operational constraints. For example, the method in \cite{donti_dc3_2021} ensures feasibility but not optimality and cannot be applied to discrete problems, while extensions to mixed-integer nonlinear programming fail to guarantee either feasibility or optimality \cite{tang_learning_2025}. Self-supervised methods such as Primal-Dual Learning (PDL) \cite{park_self-supervised_2022} and deep Lagrangian dual networks \cite{fioretto_predicting_2019} improve feasibility by integrating dual information of the underlying optimization. However, PDL relies on iterative dual updates and may struggle with nonconvex or stochastic problems, whereas the approach in \cite{fioretto_predicting_2019}, which predicts solutions to the AC Optimal Power Flow problem, is designed around problem-specific structures, which prevents generalization to settings with uncertainty or combinatorial decision variables.
}
A second promising avenue is provided by NN-constrained optimization, where the learned thermal model, represented by a NN, is embedded directly into mathematical optimization as a set of constraints.
This approach merges the expressiveness of NN with the stability and decision quality of constrained optimization. 
It has been shown that, by using a feedforward NN with Rectified Linear Unit (ReLU), the constraints can be exactly encoded as a set of  Mixed-Integer Linear (MIL) equations \cite{murzakhanov_neural_2021}.

In this paper, we formulate the day-ahead HVAC Management System (MS) as an NN-constrained optimization problem. The NN, which models the building’s thermal dynamics, is encoded precisely as a set of mixed-integer linear constraints within the optimization problem, enabling the optimization model to incorporate learned dynamics while ensuring decision quality and feasibility.

\subsection{Literature Review}
Although NNs offer strong modeling capabilities, their inherent approximation errors can propagate into suboptimal energy management decisions. Traditional ML models are trained to minimize statistical metrics such as Mean Squared Error (MSE), thus overlooking how the learned model affects downstream control decisions. In contrast, Decision-Focused Learning (DFL) aligns ML training with optimization objectives to improve decision quality \cite{li_decision-oriented_2024}. 
DFL typically relies on gradient-based training methods. The core challenge in this approach is two-fold: (\textit{i}) computing the sensitivity of the optimization problem’s solution with respect to its input parameters, and (\textit{ii}) ensuring that the resulting gradients are meaningful and informative for learning \cite{mandi_decision-focused_2023}.

A foundational contribution in DFL addressed the differentiation of unconstrained problems. In \cite{kao_directed_2009}, the authors differentiated an unconstrained Quadratic Problem (QP) to train a ML model that predicts the uncertain parameters of the QP. Donti et al. extended the framework to constrained QP \cite{donti_task-based_2017} by relaxing the constraints in the objective function to compute the gradient.
Agrawal et al. \cite{agrawal_learning_2020} enable differentiation through constrained convex programs using self-homogeneous embeddings of conic problems, allowing learning of controller parameters in the objective under known, differentiable dynamics and observable system behavior.
{Constrained convex optimization can also be differentiated via implicit differentiation of its Karush-Kuhn-Tucker conditions \cite{amos_optnet_2017}. This framework has been applied to learn the parameters of building thermal dynamics models. In \cite{chen_gnu-rl_2019}, Chen et al. propose a differentiable Model Predictive Control (MPC) policy that jointly learns the thermal dynamics and optimizes HVAC operations. Even though pioneering, their approach relies on linear thermal models that yield convex optimization problems, limiting its applicability to more realistic, non-convex building dynamics. Cui et al. extended this approach to include a forecaster upstream of the convex optimization problem~\cite{cui_decision-oriented_2025}. They jointly learn the linear thermal model and a forecaster for its residuals.} Other applications include forecasting electricity prices \cite{wahdany_more_2023} and optimizing the allocation of flexibility between transmission and distribution grid assets \cite{ortmann_tuning_2024}.

In this paper, to extend DFL to the scheduling of buildings with a NN modeling the thermal dynamics, we need to differentiate through a Mixed-Integer Problem (MIP), which is particularly difficult due to the presence of (discrete) integer decision variables. As a result, the gradient of the MIP solution with respect to the MIP parameters is typically either zero or undefined, making it incompatible with standard gradient-based learning methods.
One possible approach is to replace the MIP with its continuous relaxation, allowing the use of the existing DFL methods \cite{agrawal_learning_2020}. While relaxing a MIP to its continuous counterpart reduces computational complexity, it alters the NN representation, i.e., from capturing a nonlinear hypersurface to approximating it as a polytope, compromising modeling fidelity. Another approach generates cuts to obtain a linear program that has the same solution as the original mixed-integer formulation \cite{ferber_mipaal_2020}. However, in addition to the implementation difficulty, this method is extremely burdensome, since a new set of cuts must be generated for each training instance.

To tackle these challenges and thus enable DFL for the NN-constrained HVAC planning problem, we propose to leverage Stochastic Smoothing (SS). This technique introduces random perturbations to the uncertain parameters. {Early approaches use this perturbation to smooth the "max" operator limiting its applicability to problems in which the unknown parameters are in the linear objective function \cite{tang_pyepo_2023}. We go further by extracting informative gradients from the task-specific loss using the score function of the REINFORCE algorithm \cite{silvestri_score_2024}.} 
Unlike the aforementioned methods, SS does not make any assumption about the optimization problem class, the position of the uncertain parameters (whether in the objective function and/or the constraints), and the specific form of the loss function. Therefore, SS obtains a gradient from ex-post signals directly such as the true realized cost or constraint deviation, without requiring differentiability. SS ensures high versatility of the proposed framework and is highly suitable for combinatorial optimization involving a non-differentiable dynamical system.

\subsection{Contributions}
\noindent The main contributions of this paper can be summarized as:

\begin{enumerate}
    \item We train a NN that models the building’s thermal dynamics using DFL to optimize the decision quality rather than minimize a task-agnostic statistical metric.
    We formulate an HVAC day-ahead management as a Mixed-Integer Linear Program (MILP) in which the building’s thermal dynamics are modeled by a NN and embedded directly as optimization constraints. 
    To enable effective and robust gradient-based training of the NN, we (\textit{i}) reformulate the MILP problem as an MIQP to avoid infeasibility during learning, (\textit{ii}) solve the MIQP with random perturbations applied to unknown NN parameters to encourage exploration, (\textit{iii}) estimate gradients using a score-function method that directly relates learning with decision quality.
    \item We embed a piecewise linear NN in the HVAC management system to obtain an expressive and tractable nonlinear model of the building’s thermal dynamics. We reformulate each ReLU activation as a set of Mixed-Integer Linear (MIL) constraints, providing solution quality guarantees via a Mixed-Integer Programming (MIP) optimality gap. This approach yields a trade-off between modeling accuracy of the dynamics and computational efficiency.
    \item 
    Because reformulating NNs with ReLU activations in MILP requires Big-M constants, we improve the standard formulation by adaptively tightening them.
    Specifically, instead of relying on fixed Big-M bounds, we dynamically adjust the feasible input intervals: when an input value is known in advance (i.e., it is a deterministic parameter), its interval is reduced to a single value. This adaptive strategy produces a tighter formulation than existing fixed Big-M methods \cite{ceccon_omlt_2022, zhang_augmenting_2024}, leading to more efficient optimization without sacrificing correctness. The benefit of this approach is demonstrated through both theoretical analysis (Section~\ref{sub: tightness of the NN}) and empirical validation (Section~\ref{sub: tightness results}).
\end{enumerate}

Our framework bridges the gap between ML, constrained optimization, and real-world energy management, paving the way for more intelligent and efficient HVAC control strategies. We show the effectiveness of our approach on a realistic five-zone office building located in Denver, USA. Moreover, we compare the performance of multiple thermal models: various NN architectures and lumped RC model trained on historical data, or via traditional DFL methods.

\subsection{Outline}
In Section \ref{sec: model and methods}, we describe the HVAC MS model, including the reformulation of NN as a set of mixed-integer linear constraints, and the derivation of the resulting MILP problem. In Section \ref{sec: case study}, we introduce the case study and evaluate the performance of thermodynamic models of increasing complexity, from the RC linear model to the NN model. We compare the effectiveness of training these models in ITO and DFL fashions. In addition, we analyze the impact of the proposed tight formulation of ReLU, the standard deviation of the noise, and the number of samples on our DFL strategy. Finally, we summarize the key findings and outline the potential directions for future research in Section \ref{sec: ccl}.

%
\section{Model and Methods}
\label{sec: model and methods}

\subsection{HVAC Management System}
\label{sub: HVAC Management System}

\begin{figure*}[t]
    \centering
    \includegraphics[width=0.88\linewidth]{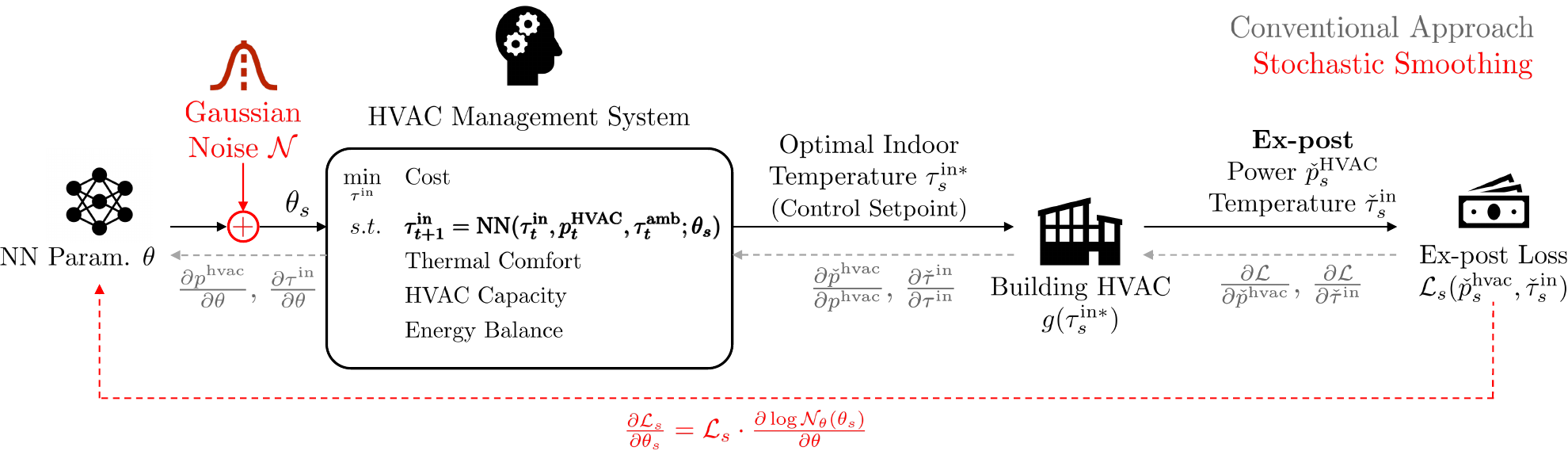}
    \caption{Comparison of the proposed stochastic smoothing DFL pipeline (red) for HVAC management system with conventional DFL approaches (gray).}
    \label{fig: DFL pipeline}
\end{figure*}
The objective of day-ahead HVAC MS is to minimize the following day's operational cost on behalf of the building manager. This scheduling task is challenging, as it must account for nonlinear building thermodynamics, efficient energy use, and sufficient thermal comfort for the occupants. In general, the easiest way to control the HVAC system is through thermostats, which are widely available in buildings. Therefore, the decision variable of the scheduling problem is the matrix of indoor temperature setpoints for each time step and thermal zone\footnote{A thermal zone in a building is a space or collection of spaces with similar space-conditioning requirements, typically sharing the same heating and cooling setpoint and controlled by a single thermostat or thermal control device.}. In the following model to optimize the day-ahead HVAC operation, we assume that the electricity cost \( \lambda_t^{\rm i} \) is lower than the feed-in tariff \(\lambda_t^{\rm e}\) for each time step.

\noindent The objective function \eqref{eq: obj fct} of the day-ahead HVAC MS is:
\begin{equation}
    \label{eq: obj fct}
    \min_{\tau^{\rm in}} \quad \underbrace{p^{\rm d} \lambda^{\rm d}}_\text{(\textit{i})} + \sum_{t=0}^{T} \left( \underbrace{p_t^{\rm i} \lambda_t^{\rm i} - p_t^{\rm e} \lambda_t^{\rm e}}_\text{(\textit{ii})} \right) \Delta t
\end{equation}
where we minimize:
\renewcommand{\labelenumi}{(\textit{\roman{enumi}})}
\begin{enumerate}
    \item the peak demand cost by applying a charge \(\lambda^{\rm d}\) to the daily power consumption peak \(p^{\rm d}\);
    \item the electricity cost aggregated at the building level which is, at each time step, the difference between the cost of consuming grid electricity \(\lambda_t^{\rm i} p_t^{\rm i}\) and the revenue from power injection into the grid \(\lambda_t^{\rm e} p_t^{\rm e}\).
\end{enumerate}

The optimization is constrained in \eqref{eq: NN for thermodynamics} by the building thermal dynamics from time steps \(t\) to \(t+1\). The relationship is modeled by a NN with parameters \(\theta\). The vector of zonal indoor temperatures \(\tau_{t+1}^{\rm in}\) is obtained based on the input from the previous time step \(t\): the vector of zonal indoor temperatures \(\tau_{t}^{\rm in} \in \mathbb{R}^Z\) where \(Z\) is the number of zones, the vector of electrical power for zonal heating \(p_{t}^{\rm h} \in \mathbb{R}^Z\), the vector of electrical power for zonal cooling \(p_{t}^{\rm c}  \in \mathbb{R}^Z\), and ambient temperature \(\tau_{t}^{\rm amb}  \in \mathbb{R}\). 
\begin{equation}
    \label{eq: NN for thermodynamics}
    \tau_{t+1}^{\rm in} = {\rm NN}(\tau_{t}^{\rm in}, p_{t}^{\rm h}, p_{t}^{\rm c}, \tau_{t}^{\rm amb}; \theta) \quad \forall t
\end{equation}

\noindent The initial state of the building, i.e., the set of initial zonal indoor temperatures, are given for all zones \(z \in \mathcal{Z}\) by:
\begin{equation}
    \label{eq: initial indoor temperature condition}
    \tau_{0, z}^{\rm in} = T_{0, z}^{\rm in}.
\end{equation}

\noindent In addition, zonal indoor temperature \(\tau_{t, z}^{\rm in}\) is restricted to lie within comfort limits \(\underline{T}_{t, z}^{\rm in}\) and \(\overline{T}_{t, z}^{\rm in}\).
\begin{equation}
    \label{eq: indoor temperature bounds}
    \underline{T}_{t, z}^{\rm in} \leq \tau_{t, z}^{\rm in} \leq \overline{T}_{t, z}^{\rm in}
\end{equation}

\noindent Heating and cooling power consumptions are positive variables bounded by the HVAC capacities \(\overline{P}_{z}^{\rm c}\) and \(\overline{P}_{z}^{\rm h}\):
\begin{flalign}
    \label{eq: zonal hvac cooling capacity}
    0 \leq p_{t, z}^{\rm c} \leq \overline{P}_{z}^{\rm c}, \quad \forall t, z; \\
    \label{eq: zonal hvac heating capacity}
    0 \leq p_{t, z}^{\rm h}  \leq \overline{P}_{z}^{\rm h}, \quad \forall t, z.
\end{flalign}

\noindent We define the zonal HVAC electrical power \(p_{t, z}^{\rm hvac}\) as the sum of the powers for cooling and heating \eqref{eq: hvac power}. The resulting energy balance is ensured by \eqref{eq: energy balance}, in which the net power consumption of the building comprises the HVAC power consumption \(p_{t, z}^{\rm hvac}\) and the non-dispatchable load \(P_t^{\rm nd}\), minus the building’s on-site power generation \(P_t^{\rm gen}\), e.g., from rooftop photovoltaic panels.
\begin{flalign}
\label{eq: hvac power}
    p_{t, z}^{\rm hvac} &= \ p_{t, z}^{\rm h} + \ p_{t, z}^{\rm c} \quad \forall t, z\\
\label{eq: energy balance}
    p_t^{\rm i} - p_t^{\rm e} &= \sum_{z=0}^{Z} p_{t, z}^{\rm hvac} + P_t^{\rm nd} - P_t^{\rm gen} \quad \forall t
\end{flalign}

\noindent The peak power demand \(p^{\rm d}\) is the maximum power consumption (or injection) from the grid over the day. The peak demand is greater than or equal to all power exchanges \( p_t^{\rm i} + p_t^{\rm e} \) across the day \eqref{eq: peak demand}. While conceptually equivalent to defining \(p^{\rm d} = \max {\{p_t^{\rm i}, p_t^{\rm e} \mid t \in \mathcal{T} \}}\), constraint \eqref{eq: peak demand} avoids introducing a bilevel structure into the optimization problem, preserving tractability.
\begin{equation}
\label{eq: peak demand}
    p^{\rm d} \geq p_t^{\rm i}  + p_t^{\rm e} \quad \forall t
\end{equation}

\noindent At all times, power imported from or exported to the grid must comply with the line capacity \(\overline{P}^{\rm l}\) of the connection between the building and the grid:
\begin{flalign}
\label{eq: import line capacity}
    p_t^{\rm i} \leq \overline{P}^{\rm l} \quad \forall t, \\
\label{eq: export line capacity}
    p_t^{\rm e} \leq \overline{P}^{\rm l} \quad \forall t.
\end{flalign}

\noindent Finally, all power imports and exports must be positive:
\begin{equation}
\label{eq: positivity}
    p_t^{\rm i},\ p_t^{\rm e} \geq 0 \qquad \forall t.
\end{equation}

\subsection{Decision-Focused Learning}
\label{sub: DFL}

\subsubsection{Problem Formulation}
Our objective is to learn the NN parameters \(\theta\) (i.e., weights and biases) of constraint \eqref{eq: NN for thermodynamics}. Instead of training \(\theta\) in a task-agnostic manner on historical data, we aim to learn \(\theta\) to directly improve HVAC management decisions as shown by \autoref{fig: DFL pipeline}.

To formalize this DFL problem, we introduce \(r\), a vector containing all stochastic parameters of the HVAC MS (e.g., weather conditions)  excluding \(\theta\). We assume \(r\) follows the known training distribution \(\mathcal{R}\). The goal is to minimize the task loss \(\mathcal{L}\), which reflects the true operational goal of the HVAC MS. This naturally leads to a bilevel optimization problem, whose solution yields the optimal values of $\theta$:
\begin{align}
    \label{eq: bilevel start}
    \min_\theta \ & \mathbb{E}_{r \sim \mathcal{R}} [\mathcal{L}(\check{p}^{\rm hvac}, \check{\tau}^{\rm in})]\\
    \label{eq: simulator}
    & {\rm s.t.} \ \check{p}^{\rm hvac}, \check{\tau}^{\rm in} = g(\tau^{\rm in*})\\
    \label{eq: inner start}
    & \quad \ \ \min f(p^{\rm hvac}, \tau^{\rm in})\\
    \label{eq: bilevel end}
    & \qquad \qquad {\rm s.t.} \ \eqref{eq: NN for thermodynamics} - \eqref{eq: positivity}
\end{align}

The inner problem \eqref{eq: inner start}-\eqref{eq: bilevel end} is the HVAC MS, including the NN formulation, which is thus a MILP. The outer problem aims to find \(\theta\) by minimizing the expected task loss over the distribution of input parameters \(r\) \eqref{eq: bilevel start}. Ideally, the task loss is the ex-post value of the decisions. Thus, it depends on the realized ex-post measurements of the HVAC power \(\check{p}^{\rm hvac}\), and  indoor temperature \(\check{\tau}^{\rm in}\). These ex-post measurements are obtained by controlling the HVAC system with optimal indoor temperature profiles \(\tau^{\rm in*}\) \eqref{eq: simulator}. In \eqref{eq: simulator}, \(g\) is the mapping between \(\tau^{\rm in*}\) and the ex-post measures \(\check{p}^{\rm hvac}\) and \(\check{\tau}^{\rm in}\). The function \(g\) models HVAC actuation, building response, and sensing.

\subsubsection{Gradient Computation}
\label{ssub: gradient computation}
Problem \eqref{eq: bilevel start}-\eqref{eq: bilevel end} is intractable and realistically large instances cannot be handled by off-the-shelf solvers. Therefore, we aim at learning \(\theta\) by gradient descent. The backpropagation step of the conventional approach applies the chain rule, but for tractability this requires \(g\) to be differentiable. In practice, real systems, and even some building simulators, are not differentiable.

To address this issue, Stochastic Smoothing (SS) aims to learn directly from an ex-post loss signal, analogous to the reward in reinforcement learning (\autoref{fig: DFL pipeline}). SS does not make any assumption about the form of the HVAC MS, the differentiability of \(g\), and the structure of the ex-post loss. To that end, instead of producing a point estimate of \(\theta\), SS models it as a distribution, e.g., a Gaussian \(\mathcal{N}(\theta, \sigma)\). Because of the stochasticity in \(\theta\), the loss \(\mathcal{L}\) in \eqref{eq: bilevel start} becomes an expectation:
\begin{equation}
    \mathcal{L} = \mathbb{E}_{\theta_s \sim \mathcal{N(\theta, \sigma)}} \left[ \mathcal{L}_s(\check{p}_s^{\rm hvac}, \check{\tau}_s^{\rm in})\right].
\end{equation}

By modeling probabilistic outputs, this approach avoids uninformative gradients that plague traditional deterministic methods. Even when the original gradient (e.g., from a LP) would be zero, taking the expectation over a smoothed distribution can yield informative, non-zero gradient signals. However, computing these gradients analytically is challenging.

To train the model, we thus employ score-function gradient estimation, a technique closely related to the \text{REINFORCE} algorithm from reinforcement learning. The approximation of the gradient becomes \cite{silvestri_score_2024}:
\begin{equation}
    \frac{\partial \mathcal{L}}{\partial \theta} = \mathbb{E}_{\theta_s \sim \mathcal{N(\theta, \sigma)}} \left[ \mathcal{L}_s \frac{\partial \log \mathcal{N}_\theta(\theta_s)}{\partial \theta}\right].
\end{equation}

In practice, the expectation of the loss \(\mathcal{L}\) is approximated using Monte Carlo with \(S\) samples \cite{mohamed_monte_2020}:
\begin{equation}
    \frac{\partial \mathcal{L}}{\partial \theta} = \frac{1}{S} \sum_{s=1}^S \left[ \mathcal{L}_s \frac{\partial \log \mathcal{N}_\theta(\theta_s)}{\partial \theta}\right].
\end{equation}

Although SS can accommodate any type of optimization problem, we reformulate the MILP as a Mixed-Integer Quadratic Program (MIQP) by relaxing the indoor temperature constraint \eqref{eq: indoor temperature bounds} and incorporating it as a quadratic penalty in the objective function \eqref{eq: obj fct relaxed}.

Reformulating the problem as a MIQP is not strictly required to obtain meaningful gradients, since SS is agnostic to the optimization structure, but it helps stabilize the learning process. By penalizing constraint violations in the objective (in the MIQP) rather than enforcing them strictly (in the original MILP), the MIQP approach prevents infeasibility during training, thereby ensuring smoother and more robust model updates. Moreover, the hard constraints can be enforced at test time, if desired.

The quadratic term (\textit{iii}) penalizes deviation of zonal indoor temperature \(\tau_{t, z}^{\rm in}\) from the target temperature \(T_{t, z}^{\rm tgt}\), reflecting the loss of thermal comfort of the occupants. The weight \(o_{t, z}\) should be designed to reflect the occupancy of the zone, i.e., higher occupancy implies a greater importance of thermal comfort and thus a stronger penalty.
\begin{flalign}
    \label{eq: obj fct relaxed}
    \min_{\tau^{\rm in}}
    p^{\rm d} \lambda^{\rm d} + \sum_{t=0}^{T} ( p_t^{\rm i} \lambda_t^{\rm i} - p_t^{\rm e} \lambda_t^{\rm e} + \underbrace{ \sum_{z=0}^{Z} o_{t, z} (\tau_{t, z}^{\rm in} - T_{t, z}^{\rm tgt})^2}_\text{(\textit{iii})})
\end{flalign}

\noindent The resulting algorithm is detailed in Algorithm \ref{alg: DFL pipeline}.

\begin{algorithm}[H]
    \caption{Algorithm for decision-focused learning of the HVAC management system.}
    \label{alg: DFL pipeline}
    \begin{algorithmic}[1]
        \State \textbf{Input:}
        \State Database of exogenous parameters \( R \)
        \State Initial NN parameters \(\theta_0\)

        \For{each epoch}
        \For{each parameter sample \(r\) in \(R\)}
        \For{each sample s}
            \State \textbf{Day-Ahead Stage:}
            \State Sample \(\theta_s \sim \mathcal{N}(\theta, \sigma)\)
            \State Solve MIQP HVAC MS problem
            
            \State \textbf{Simulator or Real Building:}
            \State Thermostat control with \(\tau_s^{\rm in*}\) as setpoints
            \State Get ex-post power \( \check{p}_s^{\rm hvac} \) and temperature \(\check{\tau}_s^{\rm in*}\)
    
            \State \textbf{Ex-Post Analysis:}
            \State Compute loss function \( \mathcal{L}_s(\check{p}_s^{\rm hvac*}, \check{\tau}_s^{\rm in*}) \)

            \EndFor
            
            \State \textbf{Backward Pass:}
            \State Compute gradient \( \frac{\partial \mathcal{L}}{\partial \theta} \approx \frac{1}{S} \sum_{s=1}^S \mathcal{L}_s \frac{\partial \log \mathcal{N}_\theta(\theta_s)}{\partial \theta} \)
            \State Update NN parameters \(\theta_{n+1} \leftarrow \theta_n - \alpha \frac{\partial \mathcal{L}}{\partial \theta}\)
        \EndFor
        \EndFor
    \end{algorithmic}
\end{algorithm}

\subsubsection{Task Loss}
The task loss \(\mathcal{L}\) aims to evaluate the quality of the decisions made by the HVAC MS. The ideal task loss is the exact ex-post value of the decisions, which includes both the cost and the thermal comfort. Here, we align our loss \(\mathcal{L}\) with the HVAC MS objective \eqref{eq: obj fct relaxed}. We name the loss \textit{Expost+} and define it as:
\begin{flalign}
    \label{eq: task loss}
    \mathcal{L} = \sum_{t=0}^{T} (\underbrace{\check{p}^{\rm hvac}_t \check{\lambda}_t + \sum_{z=0}^{Z} o_{t, z} (\Delta\check{\tau}_{t, z}^{\rm in})^2}_{\text{Ex-post cost}}) + 
    \underbrace{\vphantom{\sum_{z=0}^{Z} o_{t, z}}{\rm MSE}(\check{C}-C)}_{\text{MSE Power Cost}},
\end{flalign}
where \(\Delta\check{\tau}_{t, z}^{\rm in} = \check{\tau}_{t, z}^{\rm in} - T_{t, z}^{\rm tgt}\) and \({\rm MSE}(\check{C}-C)=\frac{1}{T}\sum_{t=0}^{T}(\check{p}^{\rm hvac} \check{\lambda}_t - p^{\rm hvac} \lambda_t)^2\).
The ex-post price \(\check{\lambda}_t\) depends on the ex-post import and export powers \(\check{p}_t^{\rm i}\) and \(\check{p}_t^{\rm e}\):
\begin{equation}
    \check{\lambda}_t = \left\{
    \begin{array}{ll}
        \lambda_t^{\rm i} \ {\rm if} \ \check{p}_t^{\rm e} = 0 \ {\rm and} \ \check{p}_t^{\rm i} \neq \check{p}_t^{\rm d}, \\
        \lambda_t^{\rm e} \ {\rm if} \ \check{p}_t^{\rm i} = 0 \ {\rm and} \ \check{p}_t^{\rm e} \neq \check{p}_t^{\rm d}, \\
        \lambda_t^{\rm i} + \lambda_t^{\rm d} \ {\rm if} \ \check{p}_t^{\rm i} = \check{p}_t^{\rm d}, \\
        \lambda_t^{\rm e} + \lambda_t^{\rm d} \ {\rm if} \ \check{p}_t^{\rm e} = \check{p}_t^{\rm d}.
    \end{array}
\right.
\end{equation}

\subsection{Neural Network MILP Formulation}
Constraint \eqref{eq: NN for thermodynamics} of the model presented in Section \ref{sub: HVAC Management System} represents the building thermal dynamics over a time interval. 
We propose learning building thermal dynamics in a data-driven fashion by leveraging piecewise linear NNs, which are obtained by designing NNs with only piecewise linear activation functions.
Specifically, we use ReLU activation functions, which are made up of two linear pieces. ReLU speeds up training and reaches greater accuracy than other activation functions for deep \cite{baldassi_properties_2019} and sparse NNs \cite{glorot_deep_2011}, but ReLU-based NNs may be nonconvex \cite{r_balestriero_mad_2021}. The ReLU function is defined as:
\begin{equation}\label{RELU}
    y={\rm max}(0, \hat y).
\end{equation}

However, embedding \eqref{RELU} directly into the larger HVAC MS optimization makes the overall problem intractable, as it requires solving a nested optimization at each ReLU. Consequently, each ReLU neuron \(n \in \mathcal{N}\) is reformulated by a binary variable \(\sigma_n\) \eqref{eq: binary relu}, two continuous variables \(y_n\) and \(\hat y_n\) \eqref{eq: y and y hat}, and a set of constraints. The complexity of the resulting optimization problem scales with the number of neurons \(N\).
\begin{flalign}
    \label{eq: binary relu}
        \sigma_{n} \in \{ 0,1 \} \\
    \label{eq: y and y hat}
        y_{n}, \hat{y}_{n} \in \mathbb{R}
\end{flalign}

\begin{figure}[t]
    \centering
    \begin{tikzpicture}[
        node distance=1cm and 1.5cm,
        every node/.style={},
        align=center]

        \def\layersepshort{2.3cm}
        \def\layerseplong{3cm}
        \def\neuronset{1.2cm}
        \def\inputx{0}
        \def\Ooney{1.8*\neuronset}
        \def\Ony{-0.35*\neuronset}
        \def\ONy{-1.6*\neuronset}

        \fill[rounded corners, fill=white!85!cyan] (-0.3*\layersepshort, 2.5*\neuronset) rectangle (0.27*\layersepshort, -2*\neuronset) {};
        \node[draw=none] at (0, 2.2*\neuronset) {Input};
        \node[draw=none, minimum size=0.5cm] (I1) at (\inputx, \Ooney) {\(x_1\)};
        \node[draw=none] at (\inputx, 0.5*\Ony+0.5*\Ooney) {\(\vdots\)};
        \node[draw=none, minimum size=0.5cm] (I2) at (\inputx, \Ony) {\(x_i\)};
        \node[draw=none] at (\inputx, 0.5*\Ony+0.5*\ONy) {\(\vdots\)};
        \node[draw=none, minimum size=0.5cm] (I3) at (\inputx, \ONy) {\(x_I\)};

        \fill[rounded corners, fill=white!85!orange] (0.3*\layersepshort, 2.5*\neuronset) rectangle (1.6*\layersepshort+\layerseplong, -2*\neuronset) {};
        \node[draw=none] at (\layersepshort+0.5*\layerseplong, 2.2*\neuronset) {ReLU Layer};

        \fill[rounded corners, fill=white!65!orange] (0.35*\layersepshort, 1.4*\neuronset) rectangle (1.55*\layersepshort+\layerseplong, -1.4*\neuronset) {};
        \node[draw=none] at (\layersepshort+0.5*\layerseplong, 1.15*\neuronset) {\small ReLU Neuron};

        \node[draw=none] at (\layersepshort+0.5*\layerseplong, 1.8*\neuronset) {\(\vdots\)};
        \node[draw=none] at (\layersepshort+0.5*\layerseplong, -1.6*\neuronset) {\(\vdots\)};

        \node[draw=none] at (1*\layersepshort, 0.15*\neuronset) {\small Preactivation};
        \node[draw, rectangle, label=above:{}] (P1) at (1*\layersepshort, -0.35*\neuronset) {\(\sum_i w_{n, i} x_i + b_n\)};
        \draw[->] (I1) -- (P1.north west);
        \draw[->] (I2) -- (P1.west);
        \draw[->] (I3) -- (P1.south west);

        \node[draw=none] at (\layersepshort+\layerseplong, 0.75*\neuronset) {\small ReLU};
        \node[draw, circle, minimum size = 2cm, clip] (A1) at (\layersepshort+\layerseplong, -0.35*\neuronset) {};
        \draw[->] (P1) -- (A1) node[midway,above,draw=none]{\(\hat y_n\)};
        \begin{scope}[shift={(\layersepshort+\layerseplong-0.84cm, -0.87*\neuronset)}, scale=0.24]
        \begin{axis}[
        axis equal,
        axis lines=middle, 
        axis line style={line width=2.5pt}, 
        xlabel={\(\hat y\)}, ylabel={\(y\)},
        xlabel style={font=\fontsize{30pt}{30pt}\selectfont}, 
        ylabel style={font=\fontsize{30pt}{30pt}\selectfont},
        samples=100, 
        domain=-2.2:2.2, 
        ymin=-0.5, ymax=2.5, 
        xticklabel={\empty}, 
        yticklabel={\empty}, 
        enlargelimits=false, 
        grid=none, 
        ]
        \addplot[line width=3pt,blue] {max(0,x)} node[above left,pos=1] {};
        \end{axis}
        \end{scope}

        \node[draw=none] at (1*\layersepshort, -1.2*\neuronset) {\small \(w_{n}, b_n \in \theta\)};
        
        \fill[rounded corners, fill=white!85!green] (1.63*\layersepshort+\layerseplong, 2.5*\neuronset) rectangle (2.2*\layersepshort+\layerseplong, -2*\neuronset) {};
        \node[draw=none] at (1.92*\layersepshort+\layerseplong, 2.2*\neuronset) {Output};
        \node[draw=none, minimum size=0.5cm] (O1) at (1.92*\layersepshort+\layerseplong, \Ooney) {\(y_1\)};
        \node[draw=none] at (1.92*\layersepshort+\layerseplong, 0.5*\Ony+0.5*\Ooney) {\(\vdots\)};
        \node[draw=none, minimum size=0.5cm] (O2) at (1.92*\layersepshort+\layerseplong, \Ony) {\(y_n\)};
        \node[draw=none] at (1.92*\layersepshort+\layerseplong, 0.5*\Ony+0.5*\ONy) {\(\vdots\)};
        \node[draw=none, minimum size=0.5cm] (O3) at (1.92*\layersepshort+\layerseplong, \ONy) {\(y_N\)};

        \draw[->] (2.74*\layersepshort, \Ooney) -- (O1);
        \draw[->] (A1) -- (O2);
        \draw[->] (2.74*\layersepshort, \ONy) -- (O3);

    \end{tikzpicture}
    \caption{Layer of neurons with Rectified Linear Unit (ReLU) activation function. 
    The parameters are the weight \(w_i\) and the bias \(b\).}
    \label{fig: relu neuron}
\end{figure}
The first step, as shown in \autoref{fig: relu neuron}, is to compute the preactivation function \eqref{eq: preactivation} of the neuron $n$ based on the input \(x\), the weights \(w_{n}\), and the bias \(b_{n}\). For \(I\) inputs, \(x \in \mathbb{R}^I\), \(w_{n} \in \mathbb{R}^{I}\), and \(b_{n} \in \mathbb{R}\). The preactivation value \(\hat y_n\) is a scalar.
\begin{equation}
    \label{eq: preactivation}
        \hat{y}_{n} = \sum_{i=1}^I w_{n,i} x_{n, i} + b_n \quad \forall n
\end{equation}

The second step is to choose a suitable formulation of the ReLU activation function, which is typically categorized into four main types: a disjunctive formulation \cite{balas_disjunctive_1979}, a Big-M formulation \cite{bunel_unified_2018}, a strong formulation \cite{anderson_strong_2020}, and a partition-based formulation \cite{tsay_partition-based_2021}. Despite offering a strong LP relaxation of the ReLU, the strong formulation requires an infinite number of constraints or a significant number of additional auxiliary variables, which hinders its performances. Here, we adopt the Big-M formulation for its simplicity and excellent empirical performance \cite{alcantara_neural_2023, kenefake_novel_2023, zhou_integrating_2023, favaro_neural_2024}.

First, the Big-M formulation of the ReLU defines the convex hull of the function as follows:
\begin{flalign}
    \label{eq: yn positive}
        y_n &\geq 0, \\[2pt]
    \label{eq: y positive}
        y_n &\geq \hat{y}_n,
\end{flalign}
where \(\hat y_n\), the preactivation value, is the input of the ReLU function and \(y_n\) is the output. The constraints \eqref{eq: ymax} and \eqref{eq: ymin} are alternatively binding depending on the value of the binary \(\sigma_n\). The parameters \(\hat{Y}_n^{\max}\) and \(\hat{Y}_n^{\min}\) are the Big-M constants. When \(\sigma_n=0\), \eqref{eq: ymax} is binding, imposing \(y_n\) to be null with \eqref{eq: yn positive}. Variable \(\hat y_n\) must be negative \eqref{eq: y positive}. In contrast, \(\sigma_n=1\) makes \eqref{eq: ymin} binding, imposing \(\hat y_n=y_n\) via \eqref{eq: y positive}. Variable \(\hat y_n\) must be positive \eqref{eq: yn positive}.
\begin{flalign}
    \label{eq: ymax}
        y_n &\le \hat{Y}_n^{\max} \cdot \sigma_n\\[2pt]
    \label{eq: ymin}
        y_n &\le \hat{y}_n - \hat{Y}_n^{\min} \cdot (1 - \sigma_n)
\end{flalign}

\subsection{Improved tightness of the NN reformulation}
\label{sub: tightness of the NN}
The Big-M constants \(\hat{Y}_n^{\min}\) and \(\hat{Y}_n^{\max}\) need to be accurately determined to produce a tight formulation of the ReLU. Three approaches exist in the literature. A first naive approach is to record the minimum and maximum values of \(\hat y_n\) during training. Although this method preserves the correlation between the NN inputs and remains simple, it is not well suited for NN-constrained optimization since physical constraints rather than the distribution of the NN training dataset should define the feasible domain \cite{favaro_neural_2024}. More advanced methods formulate optimization problems to find the bounds \(\hat{Y}_n^{\min}\) and \(\hat{Y}_n^{\max}\) \cite{tjeng_evaluating_2019, kenefake_novel_2023}. Bound optimization is computationally cumbersome, making it impractical for online computations. The third approach is to calculate the bounds in all layers of the NN based solely on the input bounds using interval analysis \cite{moore_introduction_2009}. Interval analysis enables efficient computation of bounds while ensuring physical consistency by appropriately setting the limits of the NN inputs. However, the correlation between inputs is lost.

In this work, the NN parameters \(\theta\) are updated at each gradient descent during training via DFL (cf. Section \ref{sub: DFL}), such that the bounds \(\hat{Y}_n^{\min}\) and \(\hat{Y}_n^{\max}\) must be recomputed after each gradient descent step. To handle this efficiently, we use interval analysis, and further refine the method by setting the feasible interval of NN input to singleton (i.e., degenerate) interval if the input value is known prior to solving the optimization problem. This applies to all optimization parameters that are input to the NN.

\begin{figure*}[t]
    \centering
    \begin{tikzpicture}[
        every node/.style={draw=none, font=\normalsize},
        align=center
    ]
    
        \def\seph{5cm} 
        \def\sepv{0.5cm} 
        \def\level{0cm} 

        \node[] (X) at (0, \level+\sepv) {\(x\)\\\(\in [X^{\rm min},X^{\rm max}]\)};
        \node[] (P) at (0, \level-\sepv) {\(p\)\\\(\in [P^{\rm min},P^{\rm max}]\)};

        \node[draw, rectangle, label=above:Preactivation] (PA) at (\seph, \level) {\(w_x x + w_p p + b\)};

        \node[draw, circle, label=above:ReLU] (ReLU) at (2*\seph, \level) {\(\max (0, \hat y) \)};

        \node[] (O) at (2.9*\seph, \level) {\(y \in \Delta Y \)}; 

        \draw[->] (X) -- (PA.west) node[midway,above]{\(w_x\)};
        \draw[->] (P) -- (PA.west) node[midway,below]{\(w_p\)};
        \draw[->] (PA) -- (ReLU) node[midway,above]{\(\hat y \in \Delta \hat Y\)}; 
        \draw[->] (ReLU) -- (O);
    \end{tikzpicture}
    \caption{Interval analysis of the bounds for a ReLU neuron with inputs \(p\) and \(x\), where \(p\) is a parameter and \(x\) a variable in the overarching optimization.}
    \label{fig: bounds}
\end{figure*}

\begin{theorem}[Preactivation Range Reduction]
    Consider a neuron---reformulated as mixed-integer linear constraints using Big-M constants---with preactivation \(\hat y=w_x x + w_p p + b\), where \(x \in \mathbb{R}^X\) are unknown input variables with weights \(w^x \in \mathbb{R}^{1 \times X}\), \(p \in \mathbb{R}^{P}\) are known parameters with weights \(w^p \in \mathbb{R}^{1 \times P}\), and \(b\) is the neuron bias.
    Assume that \(x\) and \(p\) are bounded such that \(x \in [X^{\rm min}, X^{\rm max}]^X\), and \(p \in [P^{\rm min}, P^{\rm max}]^P\). Let \(\Delta \hat Y\) denote the feasible range of the preactivation under these bounds. When the feasible range of \(p\) is reduced to a degenerated interval \(p \in [P_0, P_0]^P\), the feasible range of the modified preactivation \(\hat y'=w_x x + w_p P_0 + b\) satisfies:
    \begin{equation}
    \label{eq: interval ratio general app}
    \frac{|\Delta \hat Y'|}{|\Delta \hat Y|} \approx  1 - \frac{\sum_p^P w_p}{\sum_x ^X w_x + \sum_p^P w_p}
    \end{equation}
    where \(|\cdot|\) denotes the Lebesgue measure (i.e., length) of the interval.
\end{theorem}

\begin{proof}
Assuming we have a ReLU neuron with two inputs \(x \in \mathbb{R}\) and \(p \in \mathbb{R}\) that are respectively variable and parameter of the optimization problem, as illustrated by \autoref{fig: bounds}, we can demonstrate the tightening of the bounds by our approach. Let \(\Delta(\cdot)\) denote the feasible interval of its argument. If
\begin{align}
    x \in \Delta X = [X^{\rm min}, X^{\rm max}],\\
    \label{eq: parameter interval}
    p \in \Delta P = [P^{\rm min}, P^{\rm max}],
\end{align}
then, assuming without loss of generality that the associated weights \(w_x \in \mathbb{R}\) and \(w_p \in \mathbb{R}\) are positive\footnote{If a weight is negative, the lower bound of the input associated to that weight must be used in \eqref{eq: y min} and the upper bound in \eqref{eq: y max}.}, interval analysis defines the bounds of the interval of \(\hat y\) as follows
\begin{flalign}
    \label{eq: y min}
    \hat Y^{\rm min} &= w_x X^{\rm min} + w_p P^{\rm min} + b,\\
    \label{eq: y max}
    \hat Y^{\rm max} &= w_x X^{\rm max} + w_p P^{\rm max} + b.
\end{flalign}
Therefore, the interval length according to Lebesgue measure \(|\cdot|\) on \(\hat y\) is
\begin{flalign}
    |\Delta \hat Y| &= \hat Y^{\rm max} - \hat Y^{\rm min} \\
    \label{eq: y interval length}
    &= w_x  |\Delta X| + w_p |\Delta P|.
\end{flalign}
By degenerating the parameter interval \eqref{eq: parameter interval} as
\(p \in [P_0, P_0]\), where \(P_0\) is the parameter value, equation \eqref{eq: y interval length} becomes
\begin{equation}
    |\Delta \hat Y'| = w_x |\Delta X|.
\end{equation}
The ratio of the interval lengths is
\begin{flalign}
    \frac{|\Delta \hat Y'|}{|\Delta \hat Y|} &= \frac{w_x |\Delta X|}{w_x |\Delta X| + w_p |\Delta P|}, \\
    &= 1 - \frac{w_p |\Delta P|}{w_x |\Delta X| + w_p |\Delta P|}
\end{flalign}
Assuming that the inputs have been normalized, which resulted in \(|\Delta X| \approx |\Delta P|\), we have
\begin{equation}
    \label{eq: interval ratio}
    \frac{|\Delta \hat Y'|}{|\Delta \hat Y|} \approx  1 - \frac{w_p}{w_x + w_p}.
\end{equation}
Equation \eqref{eq: interval ratio} can be extended to \(P\) parameter and \(X\) variable inputs as follows:
\begin{equation}
    \frac{|\Delta \hat Y'|}{|\Delta \hat Y|} \approx  1 - \frac{\sum_p^P w_p}{\sum_x ^X w_x + \sum_p^P w_p},
\end{equation}
which proves that the proposed method is at worst equivalent to the state-of-the-art since the weights have been assumed positive. The improvement in the bound tightness depends on the ratio of the sum of parameter input weights and the sum of all weights.
\end{proof}

\noindent The bounds on \(y\) are
\begin{flalign}
    y \in [\max(0, \hat Y^{\rm min}), \max(0, \hat Y^{\rm max})].
\end{flalign}

\section{Case Study}
\label{sec: case study}

We analyze the effectiveness of the proposed method on a realistic office building comprising five zones. We start in Section~\ref{sec_data} by describing the high-quality and publicly available building model, datasets, and tools used to design the case study, ensuring a fair comparison of results and reproducibility. 
Then, Section~\ref{sec_benchmarks} presents the different benchmarks to compare the performance of our DFL-based training procedure.
In Section~\ref{sec_DFL_results}, we report the performance of all the models learned in a DFL-fashion.  
Finally, in Section~\ref{sec_ITO_results}, we further compare our DFL approach to the corresponding Identify-Then-Optimize (ITO) baseline, where the building’s thermal dynamics are learned using a standard statistical loss function independent of the optimization task.
Our code is publicly available at \href{https://github.com/PSMRB/dfl_hvac_management}{https://github.com/PSMRB/dfl\_hvac\_management}.

\subsection{Data}\label{sec_data}
The US Department of Energy has developed EnergyPlus, a high-fidelity physics-based simulator for building energy modeling, continuously updated since 2001~\cite{crawley_energyplus_2001}. EnergyPlus includes publicly available building models \cite{doe_commercial_2023}. We selected an office building with five actively controlled thermal zones. The zones are located over one floor of 511~m\textsuperscript{2}. \autoref{fig: building description} provides an illustration of the floor layout.
\begin{figure}[t]
    \centering
        \centering
        \begin{tikzpicture}[
            node distance=1cm and 1.5cm,
            every node/.style={rectangle, draw, font=\fontsize{9.5pt}{12pt}\selectfont},
            align=center]
    
            \def\dd{4cm}  
    
            \node[minimum width=1.6*\dd, minimum height=1*\dd, thick] (OW) at (0, 0) {};

            \node[minimum width=1*\dd, minimum height=0.5*\dd, dashed,
            ] (CZ) at (0, 0) {};

            \node[minimum width=1.01*\dd, minimum height=0.6*\dd, draw=none,
                    label=above:{Perimeter Zone 1, 113 m\textsuperscript{2}},
                    label=left:{Per.\\Zone 2\\67 m\textsuperscript{2}},
                    label=below:{Perimeter Zone 3, 113 m\textsuperscript{2}},
                    label=right:{Per.\\Zone 4\\67 m\textsuperscript{2}},
            ] at (0, 0.012*\dd) {Core Zone\\150 m\textsuperscript{2}};

            \draw[dashed] (OW.north east) -- (CZ.north east);  
            \draw[dashed] (OW.north west) -- (CZ.north west);  
            \draw[dashed] (OW.south west) -- (CZ.south west);  
            \draw[dashed] (OW.south east) -- (CZ.south east);  
        \end{tikzpicture}
    \caption{Layout of the building floor.}
    \label{fig: building description}
\end{figure}
Each zone is equipped with its own air-to-air heat pump, such that the HVAC system of each zone is fully independent. The only alteration to the building model is the replacement of the gas-fired heating coil by an electric coil of equal power.

The building is assumed to be located in Denver, Colorado. To model building thermal dynamics, we utilize two typical meteorological datasets, each representing a full year of hourly weather data for Denver International Airport, derived from long-term observations. The first dataset is used to simulate one year of building operation using EnergyPlus’s default heuristic control (i.e., without optimization), while the second provides weather scenarios used to evaluate optimized scheduling strategies.

Optimizing then simulating daily schedules for all 365 days is computationally demanding. Moreover, there exist days with similar weather that bring very little extra information to the learning process. Therefore, we apply \(k\)-medoids clustering to group similar weather conditions from the second dataset. The \(k\)-medoids algorithm is chosen over \(k\)-means to prevent the generation of synthetic average data. As ambient temperature is the only meteorological input required for the model \eqref{eq: NN for thermodynamics}, we use it as the primary clustering criterion. Initially, we select three fixed medoids corresponding to the most extreme conditions: the coldest day, the hottest day, and the day with the highest temperature variability. We then determine seven additional medoids to partition the rest of the dataset. The resulting temperature profiles from the \(k\)-medoids clustering are illustrated in \autoref{fig: medoids}, with the cluster means and standard deviations summarized in \autoref{fig: clusters}. In total, only ten carefully selected days are used to train the model that should generalize over the whole year.
\begin{figure}[t]
    \centering
    \begin{subfigure}[t]{0.48\textwidth}
        \centering
        \includegraphics[scale=0.20]{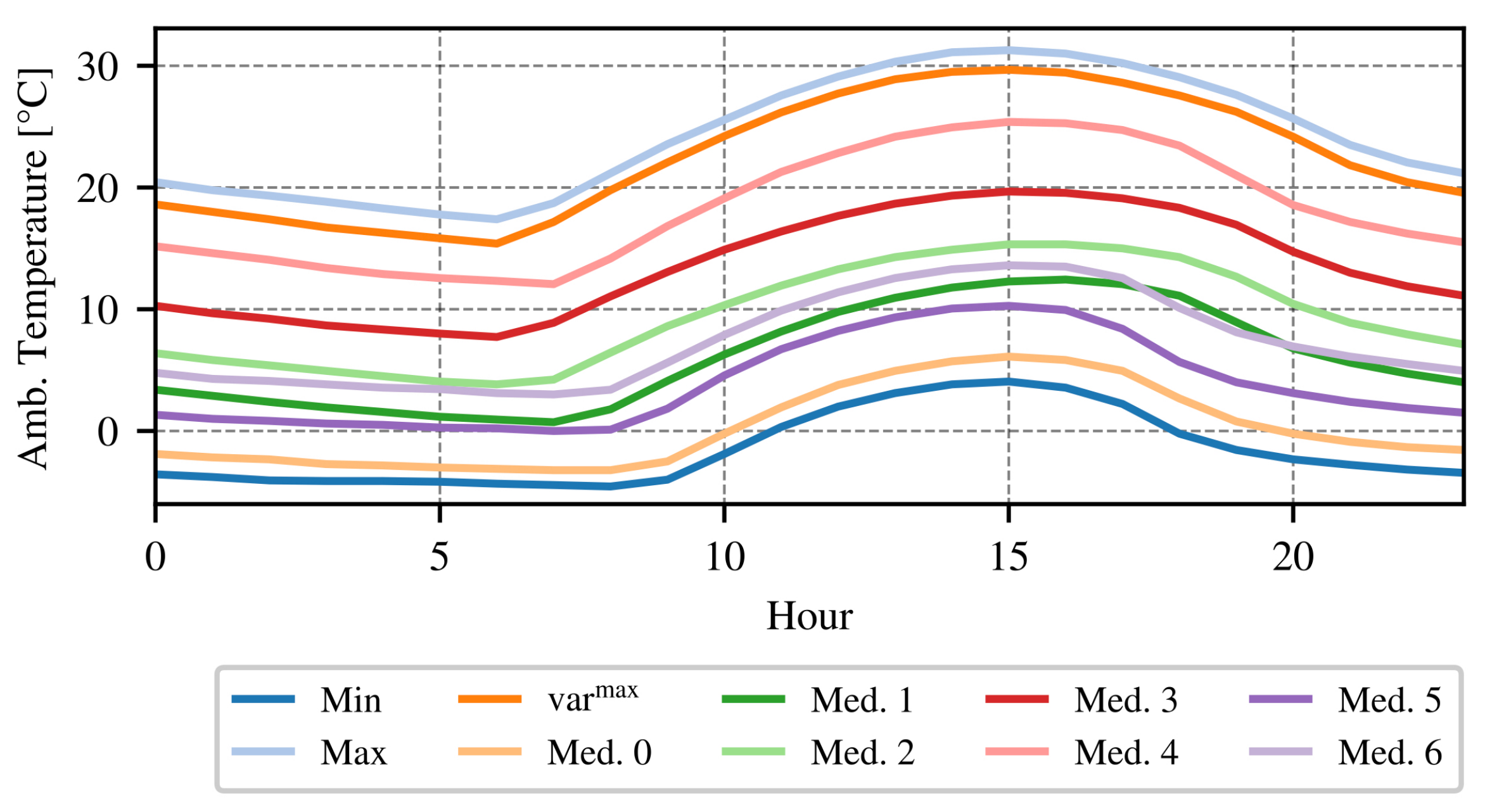}
        \caption{Ambient temperature profiles of the medoids.}
        \label{fig: medoids}
    \end{subfigure}%
    \\[1em]
    \begin{subfigure}[t]{0.48\textwidth}
        \centering
        \includegraphics[scale=0.8]{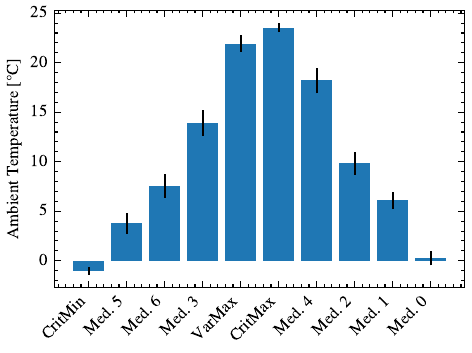}
        \caption{Mean and standard deviation of the medoids ordered to form a smooth cycle.}
        \label{fig: clusters}
    \end{subfigure}
    \caption{Cluster analysis.}
\end{figure}

To reflect typical load patterns and incentivize off-peak usage, we adopt a time-of-use electricity tariff: 0.6 \$/kWh from 6~a.m. to 7~p.m., the rate is double the base rate at 0.6 \$/kWh, and 0.3 \$/kWh during off-peak hours.


\subsection{Models and Benchmarks}
\label{sec_benchmarks}
Our goal is to learn the parameters \(\theta\) of a piecewise linear NN, i.e., constraint \eqref{eq: NN for thermodynamics}, representing the thermal model of the building. We learn the parameters in a decision-focused fashion using gradient descent. To that end, we coded in \textit{Python} and used \textit{PyTorch} for computing the backpropagation. All computations were performed on an
Intel Skylake 16-core Xeon 6142 processors and 16GB of RAM. To perform the gradient descent, we use \textit{Adam} optimizer with an initial learning rate at \(10^{-3}\) {(except for the RC model where \(10^{-2}\) was found to be better)}. The learning rate decays exponentially with \(\gamma=0.98\). We limit the number of epochs (i.e., the number of times the 10 scenarios are run) to 100 for the NNs, 150 for the RC model, and set an early stopping with a patience of 15.

We pre-train the NNs in a supervised way using the first dataset (i.e., one year of building operation with default heuristic control) to capture the building's main thermal characteristics. The resulting pre-trained NNs are then trained in a DFL fashion to learn additional and more complex relationships. This is achieved using SS.

SS learns the parameters of the NN-based thermal model embedded within an MIQP by introducing controlled random perturbations, which smooth the loss landscape and allow for gradient-based optimization. The resulting gradient is computed using the REINFORCE algorithm which enables bypassing the differentiation of the HVAC MS and EnergyPlus. {We take a single sample \(S = 1\) for the Monte Carlo approximation following conventional practice in the literature \cite{silvestri_score_2024} and allowing practical implementation with a real-life system where only one outcome can be observed. We investigated five Gaussian noises of various intensity with standard deviations of 0.01, 0.05, 0.1, \(\mu/10\), \(\mu/2\) to perturb the parameters. The value of \(\sigma\) can be fixed or variable.}

Our proposed DFL-based approach using SS is compared to other state-of-the-art DFL-based methods. These include (\textit{i}) a QP relaxation of the original optimization problem where the binaries are relaxed as continuous variables between 0 and 1 and (\textit{ii}) a two-step approach that first fixes the integer variables, then solves the resulting continuous and differentiable QP problem.
\textit{Cvxpylayer} computes the gradient of QP problems that are solved with \textit{ECOS} since it has the best performance for conic programs \cite{wahdany_more_2023}. {However, the MIQP is solved using \textit{Gurobi}.}{}

Because EnergyPlus is not differentiable, the QP relaxation and the two-step approach cannot directly optimize the ex-post value \eqref{eq: task loss}. Therefore, \cite{favaro_decision-focused_2025} introduced a \textit{supervized DFL loss} that bypasses the need to differentiate EnergyPlus (or the real installation):
\begin{flalign}
    \label{eq: loss function}
    \mathcal{L}_{\rm DFL} = \sum_{t=0}^T \frac{\check{\lambda}_t}{T}
    \left( \vphantom{\sum_{z=0}^Z {\rm MAE}}
     |p_{t}^{\rm hvac} - \check{p}_{t}^{\rm hvac}|
    + {\rm MAE}\left( p_{t, z}^{\rm hvac}\right) \right),
\end{flalign}
where \(p_{t}^{\rm hvac}\) is the total HVAC consumption of the building at time step \(t\), and \({\rm MAE}\left( p_{t, z}^{\rm hvac}\right)=\frac{1}{Z}\sum_{z=0}^Z|p_{t, z}^{\rm hvac} - \check{p}_{t, z}^{\rm hvac}|\).

Regarding the NN architectures, we learn the parameters of three fully connected feedforward NNs with six inputs: the temperature in each zone (5 inputs) and the outdoor temperature. The inputs are normalized. The three architectures have one hidden layer with 2 (NN1), 5 (NN2), and 10 ReLU (NN3). The last layer of all architectures is a linear layer with five outputs: the zonal temperatures at the next time step. Each combination of architecture and DFL method is run over five seeds, and the best outcome is reported.

Finally, we also compare the NNs to a lumped RC model, a linear model of the building based on the aggregated resistance and capacitance of each zone. For each zone, the RC model turns constraint \eqref{eq: NN for thermodynamics} into:
\begin{equation}
    \tau_{z, t+1}^{\rm in} = \tau_{z, t}^{\rm in} + \left(\frac{\eta_z^{\rm h} p_{z,t}^{\rm h} - \eta_z^{\rm c} p_{z,t}^{\rm c}}{C_z}  + \frac{\tau_{z,t}^{\rm amb} - \tau_{z,t}^{\rm in}}{R_z C_z} \right) \Delta t.
\end{equation}

\autoref{tab: hyperparameters} gathers the hyperparameters used in this case study. For more detail, you can visit the \href{https://github.com/PSMRB/dfl_hvac_management}{GitHub associated with the paper}.
\begin{table}[thb]
    \centering
    \caption{Hyperparameters.}
    \label{tab: hyperparameters}
    \begin{tabular}{l r}
        \toprule
        Optimizer & Adam \\
        Learning Rate & \(10^{-3}\) for NN \\
        & \(0.02\) for RC \\
        Exponential Decay & 0.98 \\
        Max. Nb. Epoch & 100 (150 for RC) \\
        Patience & 15 \\
        SS Noise (\(\sigma\)) & 0.01, 0.05, 0.1, $\mu$/10, $\mu$/10 \\
        S & 1 \\
        NN1 & 1 layer, 2 ReLU \\
        NN2 & 1 layer, 5 ReLU \\
        NN3 & 1 layer, 10 ReLU \\
        \bottomrule
    \end{tabular}
\end{table}

\subsection{DFL Results}\label{sec_DFL_results}
\label{sub: dfl results}

\begin{table*}[!t]
\centering
\caption{Results of the decision-focused learning models: NN with one layer of two ReLU (NN1), five ReLU (NN2), ten ReLU (NN3), and one multi-zone RC model (RC). Three methods make the MIQP problem differentiable: fixing the binaries to their optimal value (FB), relaxing the integer constraint (QP), or applying a stochatic smoothing (SS).}
\label{tab: dfl results}
\begin{tabular}{lcccccccc}
\toprule
 & \multicolumn{1}{c}{\multirow{2}{*}{RC}} & \multicolumn{3}{c}{NN1} & \multicolumn{3}{c}{NN2} & NN3  \\
 &  & \multicolumn{1}{c}{SS} & \multicolumn{1}{c}{QP} & \multicolumn{1}{c}{FB} & \multicolumn{1}{c}{SS} & \multicolumn{1}{c}{QP} & \multicolumn{1}{c}{FB} & \multicolumn{1}{c}{QP} \\ \cmidrule[0.6pt](lr){2-2} \cmidrule[0.6pt](lr){3-5} \cmidrule[0.6pt](lr){6-8} \cmidrule[0.6pt](lr){9-9}
Ex-post+ (\$) & 252 & 94 & 130 & 113 & 84 & 182 & 101 & 125 \\
Hierarchical loss & 24.7 & 20.4 & 14.3 & 12.8 & 22.2 & 12.4 & 11.4 & 13.2 \\[0.5em]
MAE (kW) & 0.6 & 0.52 & 0.55 & 0.55 & 0.5 & 0.54 & 0.54 & 0.54 \\
MSE (kW2) & 0.44 & 0.35 & 0.37 & 0.39 & 0.34 & 0.39 & 0.39 & 0.39 \\
Error mean (kW) & -0.07 & -0.08 & 0.07 & 0.04 & -0.01 & 0.02 & 0.01 & 0.03 \\
Error std (kW) & 0.39 & 0.42 & 0.37 & 0.36 & 0.45 & 0.37 & 0.39 & 0.36 \\[0.5em]
Expected cost (\$) & 32.3 & 37.9 & 45.8 & 45.8 & 39.9 & 47.4 & 43.1 & 45.4 \\
Ex-post cost (\$) & 38.0 & 41.5 & 41.1 & 41.0 & 41.8 & 41.2 & 42.7 & 43.2 \\
Cost error (\$) & 5.74 & 3.64 & -4.69 & -4.74 & 1.95 & -6.17 & -0.43 & -2.22 \\
Temp. Penalty(\$) & 162 & 31 & 57 & 41 & 35 & 59 & 47 & 58 \\[0.5em]
Nb. Epochs & 150 & 71 & 36 & 54 & 43 & 25 & 94 & 100 \\
Training time & 03:30:24 & 01:08:34 & 00:48:59 & 00:56:49 & 04:17:05 & 00:49:33 & 07:19:17 & 04:25:18 \\
Validation time & 03:14:32 & 01:07:13 & 00:44:01 & 00:52:38 & 04:22:52 & 00:59:21 & 06:16:23 & 03:54:35 \footnotemark \\
Test time & 00:00:46 & 00:00:49 & 00:01:51 & 00:00:57 & 00:07:59 & 00:01:57 & 00:01:05 & 00:10:13 \\ \bottomrule
\multicolumn{9}{l}{\small \(\vphantom{\frac{\sum_0^0}{1}}\) \(^4\)Validation conducted on the continuous relaxation (i.e., QP-relaxation) of the problem.}
\end{tabular}
\end{table*}

\autoref{tab: dfl results} reports the results of the various DFL-based models. For the two smaller NN architectures (NN1 \& NN2), the three DFL methods---i.e., SS, QP-relaxation (QP), and QP subproblem with fixed binaries (FB)---converge successfully. For the most complex architecture, NN3, only QP training is feasible within 24 hours. Indeed, the complexity of the HVAC MS increases with the number of neurons. Consequently, since the SS and FB require to solve the MIQP for each sample, the training time quickly becomes computationally intensive.

The first two rows, \textit{Ex-post+} and Hierarchical Loss, are the metrics used for training. SS can directly minimize \textit{Ex-post+} (i.e., ex-post power cost, thermal discomfort penalty, and power cost misestimation), whereas QP and FB methods minimize the error on the HVAC power through the hierarchical loss \eqref{eq: loss function}. 

SS obtains the best ex-post value compared to QP, and FB.
For NN1, SS (\(\sigma=0.01\)) obtains the lowest \textit{Ex-post+} at \$94. FB and QP follow at \$113 and \$130, respectively. All three NN1 models have an ex-post cost between \$41 and \$41.5. However, the model accuracy to predict the cost varies. SS is the most accurate model with a daily cost underestimation of \$3.7. QP and FB follow with an overestimation of \$4.69 and \$4.74, respectively.
Regarding the temperature penalties, SS comes first with only \$31 of daily penalty compared to \$41 for FB, and \$57 for QP. The total computational time (training, validation, and test) ranges from 95 minutes for QP to 136 minutes for SS.
\newline
For NN2, the same ranking is observed as for NN1. SS (\(\sigma=0.01\)) outperforms with an \textit{Ex-post+} value at \$84. This is even better than \$94 for NN1. FB and QP follow with \$101 (\$113 for NN1) and \$182 (\$130 for NN1). The ex-post cost ranges from \$41.8 for SS to \$42.7 for FB. The cost estimation error is \$1.95 for SS, but rises to \$6.17 for QP. The temperature penalty goes from \$35 for SS to \$59 and \$77 for QP and FB, respectively. In conclusion, SS offers for both architectures the best thermal comfort and with excellent estimate of the power cost along with very competitive ex-post cost and computational burden.

FB outperforms QP on almost all training and ex-post metrics. For NN1, FB has a better hierarchical loss (12.8) and \textit{Ex-post+} (\$113) than QP (hierarchical loss at 14.3 and \textit{Ex-post+} at \$130). Even though FB has a similar ex-post cost (\$41.0 versus \$41.1 for QP) and cost error (\$4.74 versus \$4.69 for QP), FB provides more thermal comfort (temperature penalty at \$41 versus \$57 for QP). FB and QP take advantage of NN2’s stronger modeling abilities to improve their training loss (FB: 11.4 vs. 12.8; QP: 12.4 vs. 14.3)). However, FB \textit{Ex-post+} falls to \$101 reflecting better ex-post metrics whereas it increases to \$182 for QP. These results show (\textit{i}) the inherent misalignment between the hierarchical loss and the ex-post value; (\textit{ii}) the poor quality of the gradient obtained by relaxing the problem.

The poor gradient approximation generated by QP even undermines the performance of larger architectures such as NN3, causing them to underperform relative to smaller SS-trained architectures. Note that for NN3, the QP relaxation is also used for validation. Since the complexity of the HVAC MS grows with the number of neurons, SS and FB become too cumbersome. Indeed, FB and SS require to solve the MIQP at each sample. QP hierarchical loss is worse for NN3 than for NN2 going up from 12.4 to 13.2, while \textit{Ex-post+} goes down from \$182 to \$125. The ex-post cost stands at \$43.2, slightly above the usual range. The cost overestimation (\$2.22) and the temperature penalty (\$58) become competitive. However, such a large NN with ten ReLU is long to train (almost four hours) and is still outperformed in terms of \textit{Ex-post+} value by NN1 trained by SS. Moreover, solving the MIQP with NN3 at test time is about ten times longer than with NN1.

In contrast, since the RC model is linear, the resulting HVAC MS is inherently a continuous QP problem. Nevertheless, the RC model cannot capture the complexity of the thermal dynamics and exhibits the worst metrics among all models. The RC model is trained to minimize the hierarchical loss, which converges at 24.7 and \textit{Ex-post+} at \$252. Regarding the ex-post metrics, the average daily ex-post cost of HVAC power is \$38.0, but it is largely underestimated at \$32.3. The average daily temperature penalty, reflecting the thermal discomfort, is \$150. Therefore, the ex-post cost is the lowest of all models because the RC model fails to provide comfortable temperatures. The total computational time for training, validation, and test is about seven hours.

In conclusion, SS outperforms DFL approaches that differentiate through the optimization problem (QP or FB), particularly when the underlying system or its simulator is non-differentiable. Unlike QP and FB, SS bypasses both the differentiation of the optimization layer and the system model, allowing the direct use of the true ex-post value (\textit{Ex-post+}) as a training signal instead of relying on surrogate task losses. This results in superior gradient estimates and more informative training loss, enabling smaller models to outperform larger QP-trained counterparts while reducing computational overhead at inference.

\subsection{ITO Results}\label{sec_ITO_results}

\begin{table*}[t]
\centering
\caption{Comparison of the performance of the Identify-Then-Optimize approach to decision-focused learning via Stochastic Smoothing (SS) for each model---NN with one layer of two ReLU neurons (NN1), five ReLU neurons (NN2), ten ReLU neurons (NN3), and one multi-zone RC model (RC).}
\label{tab: ito results}
\begin{tabular}{lcccccccc}
\toprule
 & \multicolumn{2}{c}{RC} & \multicolumn{2}{c}{NN1} & \multicolumn{2}{c}{NN2} & \multicolumn{2}{c}{NN3} \\
 & \multicolumn{1}{c}{ITO} & \multicolumn{1}{c}{DFL} & \multicolumn{1}{c}{ITO} & \multicolumn{1}{c}{SS} & \multicolumn{1}{c}{ITO} & \multicolumn{1}{c}{SS} & \multicolumn{1}{c}{ITO} & \multicolumn{1}{c}{QP} \\ \cmidrule[0.6pt](lr){2-3} \cmidrule[0.6pt](lr){4-5} \cmidrule[0.6pt](lr){6-7} \cmidrule[0.6pt](lr){8-9}
Ex-post+ (\$) & 516 & 252 & 318 & 94 & 495 & 84 & 579 & 125 \\
Hierarchical loss & 54.2 & 24.7 & 35.2 & 20.4 & 43.7 & 22.2 & 39.4 & 13.2 \\[0.5em]
MAE (kW) & 0.59 & 0.6 & 0.5 & 0.52 & 0.51 & 0.5 & 0.53 & 0.54 \\
MSE (kW2) & 0.45 & 0.44 & 0.34 & 0.35 & 0.36 & 0.34 & 0.38 & 0.39 \\
Error mean (kW) & -0.27 & -0.07 & -0.26 & -0.08 & -0.35 & -0.01 & -0.35 & 0.03 \\
Error std (kW) & 0.55 & 0.39 & 0.38 & 0.42 & 0.38 & 0.45 & 0.35 & 0.36 \\[0.5em]
Expected cost (\$) & 21.3 & 32.3 & 26.8 & 37.9 & 21.7 & 39.9 & 20.5 & 45.4 \\
Ex-post cost (\$) & 41.5 & 38.0 & 41.6 & 41.5 & 42.2 & 41.8 & 42.2 & 43.2 \\
Cost error (\$) & 20.2 & 5.7 & 14.7 & 3.6 & 20.5 & 2.0 & 21.7 & -2.2 \\
Temp. Penalty(\$) & 55 & 162 & 32 & 31 & 17 & 35 & 16 & 58 \\[0.5em]
Nb. Epochs & 0 & 150 & 0 & 71 & 0 & 43 & 0 & 100 \\
Training time & - & 03:30:24 & - & 01:08:34 & - & 04:17:05 & - & 04:25:18 \\
Validation time & - & 03:14:32 & - & 01:07:13 & - & 04:22:52 & - & 03:54:35 \\
Test time & 00:01:12 & 00:00:46 & 00:01:46 & 00:00:49 & 00:02:17 & 00:07:59 & 02:15:50 & 00:10:13 \\ \bottomrule
\end{tabular}
\end{table*}

\autoref{tab: ito results} compares the results of the DFL-trained models with their ITO counterpart. ITO models are trained over one year of historical data by Mean Square Error (MSE) minimization. The trained NN is then reformulated as a constraint of the HVAC MS. The ITO and DFL models are tested over the same ten days.

All models see a significant improvement with DFL compared to ITO.
The three NN models see major improvements in terms of \textit{Ex-post+} and hierarchical loss. \textit{Ex-post+} falls from \$318 to \$94 for NN1, from \$495 to \$84 for NN2, and from \$579 to \$125 for NN3. The cost errors of the NN1 and NN2 ITO models are \$14.7 and \$20.5, compared to \$3.6 and \$2.0 for the SS models. The QP DFL with NN3 reduces the cost error from \$21.7 to \$2.2.
Similarly, DFL of the RC parameters improves the hierarchical loss from 54.2 to 24.7 and \textit{Ex-post+} from \$516 to \$252. It shows the possible improvement DFL can bring even for simple models as already noted in \cite{favaro_decision-focused_2025}. Because of DFL, the RC model can outperform more complex ITO models such as NNs. The \textit{Ex-post+} value of the RC model is \$252, much lower than \$318 for NN1, \$495 for NN2, and \$579 for NN3.

When surveying only ITO models, NN1 is the best architecture, which shows that when ITO is used, a more complex architecture may not yield better decisions. NN1 reports an \textit{Ex-post+} ITO value at \$318 far below the RC model (\$516), NN2 (\$495), and NN3 (\$579).

In conclusion, DFL significantly improves the quality of the ex-post metrics reflecting better and more realistic decisions compared to the conventional two-stage ITO approach. A simple DFL-trained model, such as the RC model, is to be preferred to ITO-trained NNs. Nevertheless, the improvement brought by DFL is limited by the modeling power of the architecture, as shown for RC model, and the quality of gradient approximation, as shown for NN3. Caution is warranted when resorting to ITO: increasing model complexity does not necessarily yield better decisions and may introduce additional computational burden.

\subsection{Effect of Constraint Tightness on Solving Time}
\label{sub: tightness results}

\begin{table}[thb]
    \centering
    \caption{Comparison of the solving time for the optimization during test (i.e., MIQP formulation) with and without the proposed tight formulation.}
    \label{tab: tightness results}
    \addtolength{\tabcolsep}{-0.1em}
    \begin{tabular}{lcccccccr}
    \toprule
     & \multicolumn{3}{c}{NN1} & \multicolumn{3}{c}{NN2} &     NN3 &   \multirow{2}{*}{Avg.} \\
     &    \multicolumn{1}{c}{SS} & \multicolumn{1}{c}{QP} & \multicolumn{1}{c}{FB} & \multicolumn{1}{c}{SS} & \multicolumn{1}{c}{QP} & \multicolumn{1}{c}{FB} & \multicolumn{1}{c}{QP} & \\
    \cmidrule[0.6pt](lr){2-4} \cmidrule[0.6pt](lr){5-7} \cmidrule[0.6pt](lr){8-8} \cmidrule[0.6pt](lr){9-9}
    SOTA     &  39.0 &  23.4 &  23.0 &  41.5 &  72.6 &  211 &  1443 &  264.4 \\
    Tight    &  32.0 &  23.2 &  19.8 &  38.2 &  66.0 &  180 &  1422 &  254.0 \\
    Gain (\%) &  17.9 &   0.9 &  13.9 &   8.0 &   9.1 &   14.7 &     1.5 &    9.4 \\
    \bottomrule
    \end{tabular}
\end{table}

In subsection \ref{sub: tightness of the NN}, we analytically showed how to tighten the MIL equations of a ReLU by dynamically adjusting the feasible interval of its inputs.
In this case study, the NNs have 11 inputs, among which one is a parameter (i.e., the ambient temperature) and ten are decision variables (i.e., the five zonal indoor temperatures and the related HVAC power consumption). Applying \eqref{eq: interval ratio general app} and assuming that the weights have the same magnitude, the expected theoretical improvement in tightness is approximately 9.1\%.

As shown in \autoref{tab: tightness results}, the tight formulation consistently reduces solving time across all test cases, with an average gain of 9.4\%. This validates the theoretical prediction from \eqref{eq: interval ratio general app} and demonstrates the practical advantage of our approach in terms of computational efficiency. Importantly, this improvement is achieved without sacrificing model accuracy or solution quality.

\subsection{Noise intensity analysis}
In this section, we analyze the impact of the noise intensity on the performance of SS-DFL. Since we use a Gaussian noise, we modify \(\sigma\) to investigate five level of noise: 0.01, 0.05, 0.1, \(\mu/10\), and \(\mu/2\). We further enable two modes for \(\sigma\): it can either remain constant or be learned. Note that for \(\mu/10\) and \(\mu/2\), \(\sigma\) is either initialized as a fraction of \(\mu\) and then kept constant, or evolves with \(\mu\). We focus on two key metrics\footnote{The complete tables are available in the \href{https://github.com/PSMRB/dfl_hvac_management/tree/main/output/tables}{repository}.}---the \textit{Ex-post+} loss and the number of epochs necessary to converge---shown in \autoref{fig: std analysis}.

\begin{figure}[htbp]
  \centering
  \includegraphics[width=0.48\textwidth]{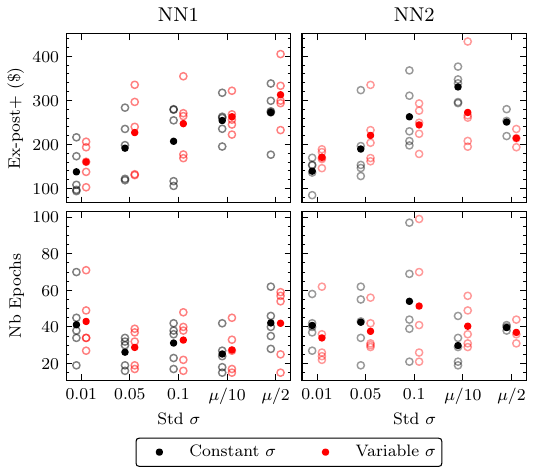}
  
  \caption{Analysis of \textit{Ex-post+} and the number of epochs for various definition of \(\sigma\). The mean of the five trainings is in bold. The \textit{Ex-post+} loss tends to increase with higher noises while no clear trend is observe on the number of epochs indicating little effect on the convergence rate.}
  \label{fig: std analysis}
\end{figure}

The analysis of noise impact on SS-DFL training performance reveals that the \textit{Ex-post+} loss increases with the noise. This is true for both NN1 and NN2, whether \(\sigma\) is constant or variable. Interestingly, setting \(\sigma\) as a fraction of \(\mu\) (i.e., \(\sigma\) is a fraction of the pre-training value found for each parameter) does not yield better results. Allowing \(\sigma\) to be trainable (variable configuration) does not improve the \textit{Ex-post+} loss. Unlike \(\mu\), \(\sigma\) has no clear optimal value. It is a variable without physical reality introduced by SS. Therefore, making \(\sigma\) a parameter to learn complexifies training without a clear gain. Ultimately, it results in a loss of performance.

In both variable and constant configurations, the best performance is achieved with the smallest constant noise (\(\sigma=0.01\)), and initializing all the \(\sigma\)'s at the same value is to be preferred to fraction of \(\mu\).

There is no significant effect of \(\sigma\) on the training convergence speed, as measured by the number of epochs required. Additionally, whether \(\sigma\) is kept constant or treated as a trainable parameter does not appear to influence convergence rates. This behavior is consistent across both NN1 and NN2 architectures.

\subsection{Number of Samples Analysis}
We complement our results by analyzing the impact of the sample size parameter \( S \) on performance, noting that all previous experiments were conducted with \( S = 1 \). The motivation for \( S = 1 \) was threefold. First, it reflects practical deployment conditions in real buildings, where only a single true realization of the building's response is observable at each time step due to the absence of simulators. Second, the score-function (REINFORCE) gradient estimator remains unbiased regardless of \( S \). Increasing \( S \) decreases the variance of the gradient estimate but increases computational cost per training instance, leading to longer training times. Third, this choice aligns with the reinforcement learning literature (and the advice of the SS-DFL authors \cite{silvestri_score_2024}), where policy gradient methods typically use single-sample estimates at each update to balance computational efficiency and estimator variance. Therefore, fixing \( S = 1 \) provides a practical and theoretically sound approach for our decision-focused learning framework.

In \autoref{fig: s analysis}, we can see the \textit{Ex-post+} loss, the training time and the number of epochs for \(S=1, 2, 5, 10\)\footnote{The complete tables are available in the \href{https://github.com/PSMRB/dfl_hvac_management}{repository}.}. All trainings were performed with \(\sigma\) set to its optimal value (i.e., constant \(\sigma=0.01\)).The \textit{Ex-post+} loss tends to increase with more samples for both NN1 and NN2. We believe this is because having more samples reduces the exploration during training. As expected, the number of epochs necessary for the training to converge decreases with \(S\), reflecting the higher quality of the gradient estimation. This decrease is extremely steady for NN1, but less for NN2. Despite the number of epochs going down with \(S\), the training time goes up as expected.

In conclusion, \(S=1\) leads to the best decision quality (i.e., minimal \textit{Ex-post+} loss) and training time.

\begin{figure}
    \centering
    \includegraphics[width=0.48\textwidth]{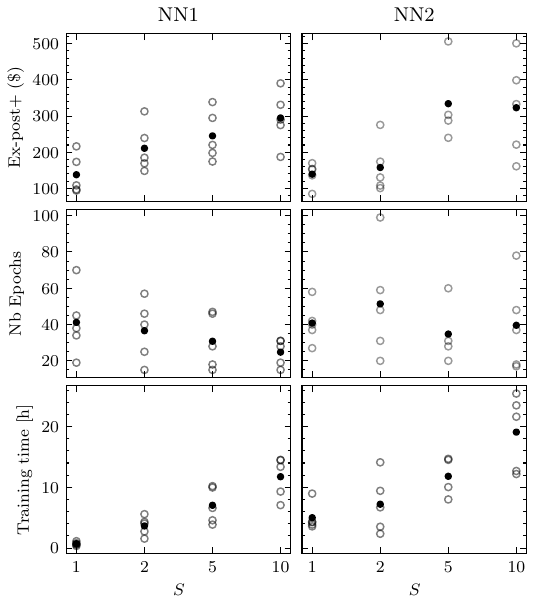}
    \caption{Analysis of \textit{Ex-post+}, the number of epochs, and the training time depending on the number of samples \(S\). The mean of the five trainings is in bold. 
    The \textit{Ex-post+} loss tends to increase with more samples. Despite the number of epochs going down with \(S\), the training time goes up.}
    \label{fig: s analysis}
\end{figure}

\section{Conclusion}
\label{sec: ccl}
We presented an HVAC MS where the building thermal dynamics are modeled using NN. We started by improving the formulation of NN as constraints in optimization problem. Then, we learned the parameters of the NN using DFL. In order to ensure meaningful gradient and training robustness, the HVAC MS is formulated as an MIQP where the thermal comfort is a quadratic penalty in the objective function rather than a hard constraint. Since MIQP are discrete by nature, we employed SS to produce an informative gradient without having to differentiate the MIQP and the building (or its simulator).

We tested our approach on a realistic five-zone building. Results show that SS outperforms the conventional two-stage approach where NN are trained on historical data and then embedded into the optimization. We also showed that DFL with SS is to be preferred to more naive approaches that relax or fix the binaries of the mixed integer problem.

{Further research is required to strengthen both scalability and robustness. On the modeling side, developing optimization formulations for entire neural networks, rather than individual neurons, and generating feasible distributions over large sets of constraint parameters would enhance stochastic smoothing and enable chance-constrained programming to address model uncertainty. On the computational side, software advances, such as parallelization, warm-starting optimization with its closest neighbor, and multi-fidelity simulation, will further reduce computation times and support large-scale deployment.}{}

\ifCLASSOPTIONcaptionsoff
  \newpage
\fi

\bibliographystyle{IEEEtran}
\bibliography{ref}

@article{afroz_modeling_2017,
  title = {Modeling Techniques Used in Building {{HVAC}} Control Systems: {{A}} Review},
  author = {Afroz, Zakia and Shafiullah, {\relax Gm}. and Urmee, Tania and Higgins, Gary},
  year = {2017},
  month = dec,
  journal = {Renewable \& Sustainable Energy Reviews},
  volume = {83},
  pages = {64--84},
  doi = {10.1016/j.rser.2017.10.044}
}

@inproceedings{agrawal_learning_2020,
  title = {Learning {{Convex Optimization Control Policies}}},
  booktitle = {Proceedings of the 2nd {{Conference}} on {{Learning}} for {{Dynamics}} and {{Control}}},
  author = {Agrawal, Akshay and Barratt, Shane and Boyd, Stephen and Stellato, Bartolomeo},
  year = {2020},
  month = jul,
  pages = {361--373},
  publisher = {PMLR},
  issn = {2640-3498},
  url = {https://proceedings.mlr.press/v120/agrawal20a.html},
  urldate = {2024-06-24},
  langid = {english}
}

@article{alcantara_neural_2023,
  title = {A Neural Network-Based Distributional Constraint Learning Methodology for Mixed-Integer Stochastic Optimization},
  author = {Alc{\'a}ntara, Antonio and Ruiz, Carlos},
  year = {2023},
  month = dec,
  journal = {Expert Systems with Applications},
  volume = {232},
  pages = {120895},
  issn = {0957-4174},
  doi = {10.1016/j.eswa.2023.120895},
  url = {https://www.sciencedirect.com/science/article/pii/S0957417423013970}
}

@article{anderson_strong_2020,
  title = {Strong Mixed-Integer Programming Formulations for Trained Neural Networks},
  author = {Anderson, Ross and Huchette, Joey and Ma, Will and Tjandraatmadja, Christian and Vielma, Juan Pablo},
  year = {2020},
  month = sep,
  journal = {Mathematical Programming},
  volume = {183},
  number = {1},
  pages = {3--39},
  issn = {1436-4646},
  doi = {10.1007/s10107-020-01474-5},
  url = {https://doi.org/10.1007/s10107-020-01474-5},
  urldate = {2024-02-15},
  langid = {english}
}

@article{balas_disjunctive_1979,
  title = {Disjunctive Programming},
  author = {Balas, Egon},
  year = {1979},
  journal = {Annals of discrete mathematics},
  volume = {5},
  pages = {3--51},
  publisher = {Elsevier},
  issn = {0167-5060}
}

@article{baldassi_properties_2019,
  title = {Properties of the {{Geometry}} of {{Solutions}} and {{Capacity}} of {{Multilayer Neural Networks}} with {{Rectified Linear Unit Activations}}},
  author = {Baldassi, Carlo and Malatesta, Enrico M. and Zecchina, Riccardo},
  year = {2019},
  journal = {Physical review letters},
  volume = {123},
  number = {17},
  pages = {1--170602},
  publisher = {Amer Physical Soc},
  address = {COLLEGE PK},
  issn = {0031-9007},
  doi = {10.1103/PhysRevLett.123.170602}
}

@inproceedings{bunel_unified_2018,
  title = {A {{Unified View}} of {{Piecewise Linear Neural Network Verification}}},
  booktitle = {Advances in {{Neural Information Processing Systems}}},
  author = {Bunel, Rudy R and Turkaslan, Ilker and Torr, Philip and Kohli, Pushmeet and Mudigonda, Pawan K},
  year = {2018},
  volume = {31},
  publisher = {Curran Associates, Inc.},
  url = {https://proceedings.neurips.cc/paper_files/paper/2018/hash/be53d253d6bc3258a8160556dda3e9b2-Abstract.html},
  urldate = {2025-02-14}
}

@misc{ceccon_omlt_2022,
  title = {{{OMLT}}: {{Optimization}} \& {{Machine Learning Toolkit}}},
  shorttitle = {{{OMLT}}},
  author = {Ceccon, Francesco and Jalving, Jordan and Haddad, Joshua and Thebelt, Alexander and Tsay, Calvin and Laird, Carl D. and Misener, Ruth},
  year = {2022},
  month = nov,
  number = {arXiv:2202.02414},
  eprint = {2202.02414},
  primaryclass = {stat},
  publisher = {arXiv},
  doi = {10.48550/arXiv.2202.02414},
  url = {http://arxiv.org/abs/2202.02414},
  urldate = {2025-03-17},
  archiveprefix = {arXiv}
}

@article{crawley_energyplus_2001,
  title = {{{EnergyPlus}}: {{Creating}} a {{New-Generation Building Energy Simulation Program}}},
  author = {Crawley, Drury and Lawrie, Linda and Winkelmann, Frederick and Buhl, W.F. and Huang, Y.Joe and Pedersen, Curtis and Strand, Richard and Liesen, Richard and Fisher, Daniel and Witte, Michael and Glazer, Jason},
  year = {2001},
  month = apr,
  journal = {Energy and Buildings},
  volume = {33},
  pages = {319--331},
  doi = {10.1016/S0378-7788(00)00114-6}
}

@article{dagostino_experimental_2022,
  title = {Experimental {{Study}} on the {{Performance Decay}} of {{Thermal Insulation}} and {{Related Influence}} on {{Heating Energy Consumption}} in {{Buildings}}},
  author = {D'Agostino, Diana and Landolfi, Roberto and Nicolella, Maurizio and Minichiello, Francesco},
  year = {2022},
  month = jan,
  journal = {Sustainability},
  volume = {14},
  number = {5},
  pages = {2947},
  publisher = {Multidisciplinary Digital Publishing Institute},
  issn = {2071-1050},
  doi = {10.3390/su14052947},
  url = {https://www.mdpi.com/2071-1050/14/5/2947},
  urldate = {2024-06-24},
  copyright = {http://creativecommons.org/licenses/by/3.0/},
  langid = {english}
}

@misc{doe_commercial_2023,
  title = {Commercial {{Prototype Building Models}}},
  author = {DOE and PNNL},
  year = {2023},
  journal = {Office of Energy Efficiency \& Renewable Energy},
  url = {https://www.energycodes.gov/development/commercial/prototype_models},
  urldate = {2025-03-07}
}

@inproceedings{donti_task-based_2017,
  title = {Task-Based End-to-End Model Learning in Stochastic Optimization},
  booktitle = {Proceedings of the 31st {{International Conference}} on {{Neural Information Processing Systems}}},
  author = {Donti, Priya L. and Amos, Brandon and Kolter, J. Zico},
  year = {2017},
  month = dec,
  series = {{{NIPS}}'17},
  pages = {5490--5500},
  publisher = {Curran Associates Inc.},
  address = {Red Hook, NY, USA},
  urldate = {2024-06-24},
  isbn = {978-1-5108-6096-4}
}

@inproceedings{favaro_decision-focused_2025,
  title = {Decision-{{Focused Learning}} for {{Complex System Identification}}: {{HVAC Management System Application}}},
  shorttitle = {Decision-{{Focused Learning}} for {{Complex System Identification}}},
  booktitle = {Proceedings of the 16th {{ACM International Conference}} on {{Future}} and {{Sustainable Energy Systems}}},
  author = {Favaro, Pietro and Toubeau, Jean-Fran{\c c}ois and Vall{\'e}e, Fran{\c c}ois and Dvorkin, Yury},
  year = {2025},
  month = jun,
  series = {E-{{Energy}} '25},
  pages = {347--358},
  publisher = {Association for Computing Machinery},
  address = {New York, NY, USA},
  doi = {10.1145/3679240.3734584},
  url = {https://dl.acm.org/doi/10.1145/3679240.3734584},
  urldate = {2025-06-24},
  isbn = {979-8-4007-1125-1}
}

@article{favaro_neural_2024,
  title = {Neural Network Informed Day-Ahead Scheduling of Pumped Hydro Energy Storage},
  author = {Favaro, Pietro and Dol{\'a}nyi, Mih{\'a}ly and Vall{\'e}e, Fran{\c c}ois and Toubeau, Jean-Fran{\c c}ois},
  year = {2024},
  month = feb,
  journal = {Energy},
  volume = {289},
  pages = {129999},
  issn = {0360-5442},
  doi = {10.1016/j.energy.2023.129999},
  url = {https://www.sciencedirect.com/science/article/pii/S0360544223033935}
}

@article{ferber_mipaal_2020,
  title = {{{MIPaaL}}: {{Mixed Integer Program}} as a {{Layer}}},
  author = {Ferber, Aaron and Wilder, Bryan and Dilkina, Bistra and Tambe, Milind},
  year = {2020},
  month = apr,
  journal = {Proceedings of the AAAI Conference on Artificial Intelligence},
  volume = {34},
  number = {02},
  pages = {1504--1511},
  doi = {10.1609/aaai.v34i02.5509},
  url = {https://ojs.aaai.org/index.php/AAAI/article/view/5509},
  urldate = {2025-02-20},
  chapter = {AAAI Technical Track: Constraint Satisfaction and Optimization}
}

@article{glorot_deep_2011,
  title = {Deep Sparse Rectifier Neural Networks},
  author = {Glorot, Xavier and Bordes, Antoine and Bengio, Yoshua},
  year = {2011},
  journal = {Journal of machine learning research},
  volume = {15},
  pages = {315--323},
  issn = {1532-4435}
}

@inproceedings{kao_directed_2009,
  title = {Directed {{Regression}}},
  booktitle = {Neural {{Information Processing Systems}}},
  author = {Kao, Yi-Hao and Roy, Benjamin Van and Yan, X.},
  year = {2009},
  month = dec,
  url = {https://www.semanticscholar.org/paper/Directed-Regression-Kao-Roy/76efadb1274df410d7060ca0b184c4dba354e96c},
  urldate = {2024-06-24}
}

@article{kenefake_novel_2023,
  title = {A Novel Neural Network Bounds-Tightening Procedure for Multiparametric Programming and Control},
  author = {Kenefake, Dustin and Kakaodkar, Rahul and Ali, Moustafa and Pistikopoulos, Efstratios N.},
  editor = {Kokossis, Antonios C. and Georgiadis, Michael C. and Pistikopoulos, Efstratios},
  year = {2023},
  month = jan,
  journal = {Computer Aided Chemical Engineering},
  volume = {52},
  pages = {1841--1846},
  doi = {10.1016/B978-0-443-15274-0.50292-4},
  url = {https://www.sciencedirect.com/science/article/pii/B9780443152740502924}
}

@article{kim_experimental_2016,
  title = {Experimental {{Study}} of {{Grid Frequency Regulation Ancillary Service}} of a {{Variable Speed Heat Pump}}},
  author = {Kim, Young-Jin and Fuentes, Elena and Norford, Leslie K.},
  year = {2016},
  month = jul,
  journal = {IEEE Transactions on Power Systems},
  volume = {31},
  number = {4},
  pages = {3090--3099},
  issn = {1558-0679},
  doi = {10.1109/TPWRS.2015.2472497},
  url = {https://ieeexplore.ieee.org/abstract/document/7268771},
  urldate = {2025-03-13}
}

@article{kouzelis_estimation_2015,
  title = {Estimation of {{Residential Heat Pump Consumption}} for {{Flexibility Market Applications}}},
  author = {Kouzelis, Konstantinos and Tan, Zheng H. and {Bak-Jensen}, Birgitte and Pillai, Jayakrishnan Radhakrishna and Ritchie, Ewen},
  year = {2015},
  month = jul,
  journal = {IEEE Transactions on Smart Grid},
  volume = {6},
  number = {4},
  pages = {1852--1864},
  issn = {1949-3061},
  doi = {10.1109/TSG.2015.2414490},
  url = {https://ieeexplore.ieee.org/abstract/document/7079500},
  urldate = {2025-03-13}
}

@misc{li_decision-oriented_2024,
  title = {Decision-{{Oriented Learning}} for {{Future Power System Decision-Making}} under {{Uncertainty}}},
  author = {Li, Ran and Zhang, Haipeng and Sun, Mingyang and Teng, Fei and Wan, Can and Pineda, Salvador and Kariniotakis, Georges},
  year = {2024},
  month = apr,
  number = {arXiv:2401.03680},
  eprint = {2401.03680},
  primaryclass = {eess},
  publisher = {arXiv},
  doi = {10.48550/arXiv.2401.03680},
  url = {http://arxiv.org/abs/2401.03680},
  urldate = {2025-03-13},
  archiveprefix = {arXiv}
}

@article{m_yousefi_predictive_2021,
  title = {Predictive {{Home Energy Management System With Photovoltaic Array}}, {{Heat Pump}}, and {{Plug-In Electric Vehicle}}},
  author = {{M. Yousefi} and {A. Hajizadeh} and {M. N. Soltani} and {B. Hredzak}},
  year = {2021},
  month = jan,
  journal = {IEEE Transactions on Industrial Informatics},
  volume = {17},
  number = {1},
  pages = {430--440},
  issn = {1941-0050},
  doi = {10.1109/TII.2020.2971530}
}

@misc{mandi_decision-focused_2023,
  title = {Decision-{{Focused Learning}}: {{Foundations}}, {{State}} of the {{Art}}, {{Benchmark}} and {{Future Opportunities}}},
  shorttitle = {Decision-{{Focused Learning}}},
  author = {Mandi, Jayanta and Kotary, James and Berden, Senne and Mulamba, Maxime and Bucarey, Victor and Guns, Tias and Fioretto, Ferdinando},
  year = {2023},
  month = aug,
  number = {arXiv:2307.13565},
  eprint = {2307.13565},
  primaryclass = {cs, math},
  publisher = {arXiv},
  doi = {10.48550/arXiv.2307.13565},
  url = {http://arxiv.org/abs/2307.13565},
  urldate = {2024-02-08},
  archiveprefix = {arXiv}
}

@inproceedings{mocanu_energy_2016,
  title = {Energy Disaggregation for Real-Time Building Flexibility Detection},
  booktitle = {2016 {{IEEE Power}} and {{Energy Society General Meeting}} ({{PESGM}})},
  author = {Mocanu, Elena and Nguyen, Phuong H and Gibescu, Madeleine},
  year = {2016},
  pages = {1--5},
  publisher = {IEEE},
  isbn = {1-5090-4168-0}
}

@article{mohamed_monte_2020,
  title = {Monte Carlo Gradient Estimation in Machine Learning},
  author = {Mohamed, Shakir and Rosca, Mihaela and Figurnov, Michael and Mnih, Andriy},
  year = {2020},
  journal = {Journal of Machine Learning Research},
  volume = {21},
  number = {132},
  pages = {1--62},
  isbn = {1533-7928}
}

@book{moore_introduction_2009,
  title = {Introduction to {{Interval Analysis}}},
  author = {Moore, Ramon E. and Kearfott, R. Baker and Cloud, Michael J.},
  year = {2009},
  month = jan,
  publisher = {{Society for Industrial and Applied Mathematics}},
  doi = {10.1137/1.9780898717716},
  url = {http://epubs.siam.org/doi/book/10.1137/1.9780898717716},
  urldate = {2025-02-18},
  isbn = {978-0-89871-669-6 978-0-89871-771-6},
  langid = {english}
}

@article{murzakhanov_neural_2021,
  title = {Neural {{Networks}} for {{Encoding Dynamic Security-Constrained Optimal Power Flow}}},
  author = {Murzakhanov, Ilgiz and Venzke, Andreas and Misyris, George S. and Chatzivasileiadis, Spyros},
  year = {2021},
  month = oct,
  journal = {arXiv:2003.07939 [cs, eess, math]},
  eprint = {2003.07939},
  primaryclass = {cs, eess, math},
  url = {http://arxiv.org/abs/2003.07939},
  urldate = {2022-01-25},
  archiveprefix = {arXiv}
}

@misc{noauthor_buildings_2023,
  title = {Buildings - {{Energy System}}},
  year = {2023},
  month = jul,
  journal = {IEA},
  url = {https://www.iea.org/energy-system/buildings},
  urldate = {2025-03-11},
  langid = {british}
}

@misc{noauthor_current_2015,
  title = {The Current Electricity Market Design in {{Europe}}},
  year = {2015},
  month = jan,
  publisher = {KU Leuven Energy Institute},
  url = {https://set.kuleuven.be/ei/factsheets},
  urldate = {2022-04-14},
  langid = {english}
}

@misc{noauthor_energy_2024,
  title = {Energy Consumption in Households},
  year = {2024},
  month = jun,
  journal = {Eurostat},
  url = {https://ec.europa.eu/eurostat/statistics-explained/index.php?title=Energy_consumption_in_households},
  urldate = {2025-03-17},
  langid = {english}
}

@misc{ortmann_tuning_2024,
  title = {Tuning and {{Testing}} an {{Online Feedback Optimization Controller}} to {{Provide Curative Distribution Grid Flexibility}}},
  author = {Ortmann, Lukas and B{\"o}hm, Fabian and {Klein-Helmkamp}, Florian and Ulbig, Andreas and Bolognani, Saverio and D{\"o}rfler, Florian},
  year = {2024},
  month = mar,
  number = {arXiv:2403.01782},
  eprint = {2403.01782},
  primaryclass = {cs, eess},
  publisher = {arXiv},
  url = {http://arxiv.org/abs/2403.01782},
  urldate = {2024-06-21},
  archiveprefix = {arXiv},
  langid = {english}
}

@article{papadaskalopoulos_decentralized_2013,
  title = {Decentralized {{Participation}} of {{Flexible Demand}} in {{Electricity Markets}}---{{Part II}}: {{Application With Electric Vehicles}} and {{Heat Pump Systems}}},
  shorttitle = {Decentralized {{Participation}} of {{Flexible Demand}} in {{Electricity Markets}}---{{Part II}}},
  author = {Papadaskalopoulos, Dimitrios and Strbac, Goran and Mancarella, Pierluigi and Aunedi, Marko and Stanojevic, Vladimir},
  year = {2013},
  month = nov,
  journal = {IEEE Transactions on Power Systems},
  volume = {28},
  number = {4},
  pages = {3667--3674},
  issn = {1558-0679},
  doi = {10.1109/TPWRS.2013.2245687},
  url = {https://ieeexplore.ieee.org/abstract/document/6515366},
  urldate = {2025-03-13}
}

@article{qiang_improved_2015,
  title = {An Improved Office Building Cooling Load Prediction Model Based on Multivariable Linear Regression},
  author = {Qiang, Guo and Zhe, Tian and Yan, Ding and Neng, Zhu},
  year = {2015},
  month = nov,
  journal = {Energy and Buildings},
  volume = {107},
  pages = {445--455},
  issn = {0378-7788},
  doi = {10.1016/j.enbuild.2015.08.041},
  url = {https://www.sciencedirect.com/science/article/pii/S0378778815302255}
}

@article{r_balestriero_mad_2021,
  title = {Mad {{Max}}: {{Affine Spline Insights Into Deep Learning}}},
  author = {{R. Balestriero} and {R. G. Baraniuk}},
  year = {2021},
  month = may,
  journal = {Proceedings of the IEEE},
  volume = {109},
  number = {5},
  pages = {704--727},
  issn = {1558-2256},
  doi = {10.1109/JPROC.2020.3042100}
}

@article{s_a_nabavi_deep_2021,
  title = {Deep {{Learning}} in {{Energy Modeling}}: {{Application}} in {{Smart Buildings With Distributed Energy Generation}}},
  author = {{S. A. Nabavi} and {N. H. Motlagh} and {M. A. Zaidan} and {A. Aslani} and {B. Zakeri}},
  year = {2021},
  journal = {IEEE Access},
  volume = {9},
  pages = {125439--125461},
  issn = {2169-3536},
  doi = {10.1109/ACCESS.2021.3110960}
}

@misc{silvestri_score_2024,
  title = {Score {{Function Gradient Estimation}} to {{Widen}} the {{Applicability}} of {{Decision-Focused Learning}}},
  author = {Silvestri, Mattia and Berden, Senne and Mandi, Jayanta and Mahmuto{\u g}ullar{\i}, Ali Irfan and Amos, Brandon and Guns, Tias and Lombardi, Michele},
  year = {2024},
  month = jun,
  number = {arXiv:2307.05213},
  eprint = {2307.05213},
  primaryclass = {cs},
  publisher = {arXiv},
  url = {http://arxiv.org/abs/2307.05213},
  urldate = {2024-08-30},
  archiveprefix = {arXiv},
  langid = {english}
}

@misc{tang_pyepo_2023,
    title = {{PyEPO}: {A} {PyTorch}-based {End}-to-{End} {Predict}-then-{Optimize} {Library} for {Linear} and {Integer} {Programming}},
    shorttitle = {{PyEPO}},
    url = {http://arxiv.org/abs/2206.14234},
    doi = {10.48550/arXiv.2206.14234},
    urldate = {2025-08-21},
    publisher = {arXiv},
    author = {Tang, Bo and Khalil, Elias B.},
    month = apr,
    year = {2023},
    note = {arXiv:2206.14234 [math]},
}

@article{tian_real-time_2021,
  title = {Real-Time Flexibility Quantification of a Building {{HVAC}} System for Peak Demand Reduction},
  author = {Tian, Guanyu and Sun, Qun Zhou and Wang, Wenyi},
  year = {2021},
  journal = {IEEE Transactions on Power Systems},
  volume = {37},
  number = {5},
  pages = {3862--3874},
  publisher = {IEEE},
  issn = {0885-8950}
}

@misc{tjeng_evaluating_2019,
  title = {Evaluating {{Robustness}} of {{Neural Networks}} with {{Mixed Integer Programming}}},
  author = {Tjeng, Vincent and Xiao, Kai and Tedrake, Russ},
  year = {2019},
  month = feb,
  number = {arXiv:1711.07356},
  eprint = {1711.07356},
  primaryclass = {cs},
  publisher = {arXiv},
  doi = {10.48550/arXiv.1711.07356},
  url = {http://arxiv.org/abs/1711.07356},
  urldate = {2024-04-15},
  archiveprefix = {arXiv}
}

@inproceedings{tsay_partition-based_2021,
  title = {Partition-{{Based Formulations}} for {{Mixed-Integer Optimization}} of {{Trained ReLU Neural Networks}}},
  booktitle = {Advances in {{Neural Information Processing Systems}}},
  author = {Tsay, Calvin and Kronqvist, Jan and Thebelt, Alexander and Misener, Ruth},
  year = {2021},
  volume = {34},
  pages = {3068--3080},
  publisher = {Curran Associates, Inc.},
  url = {https://proceedings.neurips.cc/paper_files/paper/2021/hash/17f98ddf040204eda0af36a108cbdea4-Abstract.html},
  urldate = {2025-02-14}
}

@article{wahdany_more_2023,
  title = {More than Accuracy: End-to-End Wind Power Forecasting That Optimises the Energy System},
  shorttitle = {More than Accuracy},
  author = {Wahdany, Dariush and Schmitt, Carlo and Cremer, Jochen L.},
  year = {2023},
  month = aug,
  journal = {Electric Power Systems Research},
  volume = {221},
  pages = {109384},
  issn = {0378-7796},
  doi = {10.1016/j.epsr.2023.109384},
  url = {https://www.sciencedirect.com/science/article/pii/S0378779623002730},
  urldate = {2024-06-25}
}

@article{wang_online_2012,
  title = {Online Voltage Security Assessment Considering Comfort-Constrained Demand Response Control of Distributed Heat Pump Systems},
  author = {Wang, D. and Parkinson, S. and Miao, W. and Jia, H. and Crawford, C. and Djilali, N.},
  year = {2012},
  month = aug,
  journal = {Smart Grids},
  volume = {96},
  pages = {104--114},
  issn = {0306-2619},
  doi = {10.1016/j.apenergy.2011.12.005},
  url = {https://www.sciencedirect.com/science/article/pii/S0306261911007999}
}

@article{y_-j_kim_supervised-learning-based_2020,
  title = {A {{Supervised-Learning-Based Strategy}} for {{Optimal Demand Response}} of an {{HVAC System}} in a {{Multi-Zone Office Building}}},
  author = {{Y. -J. Kim}},
  year = {2020},
  month = sep,
  journal = {IEEE Transactions on Smart Grid},
  volume = {11},
  number = {5},
  pages = {4212--4226},
  issn = {1949-3061},
  doi = {10.1109/TSG.2020.2986539}
}

@article{zhang_augmenting_2024,
  title = {Augmenting Optimization-Based Molecular Design with Graph Neural Networks},
  author = {Zhang, Shiqiang and Campos, Juan S. and Feldmann, Christian and Sandfort, Frederik and Mathea, Miriam and Misener, Ruth},
  year = {2024},
  month = jul,
  journal = {Computers \& Chemical Engineering},
  volume = {186},
  pages = {108684},
  issn = {0098-1354},
  doi = {10.1016/j.compchemeng.2024.108684},
  url = {https://www.sciencedirect.com/science/article/pii/S0098135424001029},
  urldate = {2025-03-17}
}

@article{zhou_integrating_2023,
  title = {Integrating Machine Learning and Mathematical Programming for Efficient Optimization of Operating Conditions in Organic {{Rankine}} Cycle ({{ORC}}) Based Combined Systems},
  author = {Zhou, Jianzhao and Chu, Yin Ting and Ren, Jingzheng and Shen, Weifeng and He, Chang},
  year = {2023},
  month = oct,
  journal = {Energy},
  volume = {281},
  pages = {128218},
  issn = {0360-5442},
  doi = {10.1016/j.energy.2023.128218},
  url = {https://www.sciencedirect.com/science/article/pii/S0360544223016122}
}

@inproceedings{chen_gnu-rl_2019,
    address = {New York, NY, USA},
    series = {{BuildSys} '19},
    title = {Gnu-{RL}: {A} {Precocial} {Reinforcement} {Learning} {Solution} for {Building} {HVAC} {Control} {Using} a {Differentiable} {MPC} {Policy}},
    isbn = {978-1-4503-7005-9},
    shorttitle = {Gnu-{RL}},
    url = {https://dl.acm.org/doi/10.1145/3360322.3360849},
    doi = {10.1145/3360322.3360849},
    urldate = {2025-08-21},
    booktitle = {Proceedings of the 6th {ACM} {International} {Conference} on {Systems} for {Energy}-{Efficient} {Buildings}, {Cities}, and {Transportation}},
    publisher = {Association for Computing Machinery},
    author = {Chen, Bingqing and Cai, Zicheng and Bergés, Mario},
    month = nov,
    year = {2019},
    pages = {316--325},
}

@misc{donti_dc3_2021,
  title = {{{DC3}}: {{A}} Learning Method for Optimization with Hard Constraints},
  shorttitle = {{{DC3}}},
  author = {Donti, Priya L. and Rolnick, David and Kolter, J. Zico},
  year = {2021},
  month = apr,
  number = {arXiv:2104.12225},
  eprint = {2104.12225},
  primaryclass = {cs},
  publisher = {arXiv},
  doi = {10.48550/arXiv.2104.12225},
  url = {http://arxiv.org/abs/2104.12225},
  urldate = {2025-09-25},
  archiveprefix = {arXiv}
}

@misc{tang_learning_2025,
    title = {Learning to {Optimize} for {Mixed}-{Integer} {Non}-linear {Programming} with {Feasibility} {Guarantees}},
    url = {http://arxiv.org/abs/2410.11061},
    doi = {10.48550/arXiv.2410.11061},
    urldate = {2025-08-22},
    publisher = {arXiv},
    author = {Tang, Bo and Khalil, Elias B. and Drgoňa, Ján},
    month = may,
    year = {2025},
    note = {arXiv:2410.11061 [cs]},
}

@misc{fioretto_predicting_2019,
    title = {Predicting {AC} {Optimal} {Power} {Flows}: {Combining} {Deep} {Learning} and {Lagrangian} {Dual} {Methods}},
    shorttitle = {Predicting {AC} {Optimal} {Power} {Flows}},
    url = {http://arxiv.org/abs/1909.10461},
    doi = {10.48550/arXiv.1909.10461},
    urldate = {2025-09-04},
    publisher = {arXiv},
    author = {Fioretto, Ferdinando and Mak, Terrence W. K. and Hentenryck, Pascal Van},
    month = dec,
    year = {2019},
    note = {arXiv:1909.10461 [eess]},
}

@misc{park_self-supervised_2022,
    title = {Self-{Supervised} {Primal}-{Dual} {Learning} for {Constrained} {Optimization}},
    url = {http://arxiv.org/abs/2208.09046},
    doi = {10.48550/arXiv.2208.09046},
    urldate = {2025-09-04},
    publisher = {arXiv},
    author = {Park, Seonho and Hentenryck, Pascal Van},
    month = nov,
    year = {2022},
    note = {arXiv:2208.09046 [cs]},
}

@article{jang_active_2024,
    title = {Active {Reinforcement} {Learning} for {Robust} {Building} {Control}},
    volume = {38},
    url = {https://ojs.aaai.org/index.php/AAAI/article/view/30219},
    doi = {10.1609/aaai.v38i20.30219},
    number = {20},
    urldate = {2025-03-13},
    journal = {Proceedings of the AAAI Conference on Artificial Intelligence},
    author = {Jang, Doseok and Yan, Larry and Spangher, Lucas and Spanos, Costas J.},
    month = mar,
    year = {2024},
    note = {Section: AAAI Technical Track on AI for Social Impact Track},
    pages = {22150--22158},
}

@inproceedings{amos_optnet_2017,
    title = {{OptNet}: {Differentiable} {Optimization} as a {Layer} in {Neural} {Networks}},
    shorttitle = {{OptNet}},
    url = {https://proceedings.mlr.press/v70/amos17a.html},
    language = {en},
    urldate = {2025-01-07},
    booktitle = {Proceedings of the 34th {International} {Conference} on {Machine} {Learning}},
    publisher = {PMLR},
    author = {Amos, Brandon and Kolter, J. Zico},
    month = jul,
    year = {2017},
    note = {ISSN: 2640-3498},
    pages = {136--145},
}

@article{cui_decision-oriented_2025,
    title = {Decision-{Oriented} {Modeling} of {Thermal} {Dynamics} {Within} {Buildings}},
    volume = {16},
    issn = {1949-3061},
    url = {https://ieeexplore.ieee.org/document/10638763/},
    doi = {10.1109/TSG.2024.3445574},
    number = {1},
    urldate = {2025-08-08},
    journal = {IEEE Transactions on Smart Grid},
    author = {Cui, Xueyuan and Toubeau, Jean-François and Vallée, François and Wang, Yi},
    month = jan,
    year = {2025},
    pages = {369--382},
}

@article{drgona_all_2020,
    title = {All you need to know about model predictive control for buildings},
    volume = {50},
    issn = {1367-5788},
    url = {https://www.sciencedirect.com/science/article/pii/S1367578820300584},
    doi = {10.1016/j.arcontrol.2020.09.001},
    journal = {Annual Reviews in Control},
    author = {Drgoňa, Ján and Arroyo, Javier and Cupeiro Figueroa, Iago and Blum, David and Arendt, Krzysztof and Kim, Donghun and Ollé, Enric Perarnau and Oravec, Juraj and Wetter, Michael and Vrabie, Draguna L. and Helsen, Lieve},
    month = jan,
    year = {2020},
    pages = {190--232},
}

@article{drgona_physics-constrained_2021,
    title = {Physics-constrained deep learning of multi-zone building thermal dynamics},
    volume = {243},
    issn = {0378-7788},
    url = {https://www.sciencedirect.com/science/article/pii/S0378778821002760},
    doi = {10.1016/j.enbuild.2021.110992},
    urldate = {2025-08-08},
    journal = {Energy and Buildings},
    author = {Drgoňa, Ján and Tuor, Aaron R. and Chandan, Vikas and Vrabie, Draguna L.},
    month = jul,
    year = {2021},
    pages = {110992},
}

@article{kircher_lumped_2015,
    title = {On the lumped capacitance approximation accuracy in {RC} network building models},
    volume = {108},
    issn = {0378-7788},
    url = {https://www.sciencedirect.com/science/article/pii/S0378778815302930},
    doi = {10.1016/j.enbuild.2015.09.053},
    urldate = {2024-06-24},
    journal = {Energy and Buildings},
    author = {Kircher, Kevin J. and Max Zhang, K.},
    month = dec,
    year = {2015},
    pages = {454--462},
}

@inproceedings{f_belic_thermal_2016,
    title = {Thermal modeling of buildings with {RC} method and parameter estimation},
    doi = {10.1109/SST.2016.7765626},
    booktitle = {2016 {International} {Conference} on {Smart} {Systems} and {Technologies} ({SST})},
    author = {{F. Belić} and {Ž. Hocenski} and {D. Slišković}},
    month = oct,
    year = {2016},
    note = {Journal Abbreviation: 2016 International Conference on Smart Systems and Technologies (SST)},
    pages = {19--25},
}

\begin{IEEEbiography}
[{\includegraphics[width=1in,height=1.25in,clip,keepaspectratio]{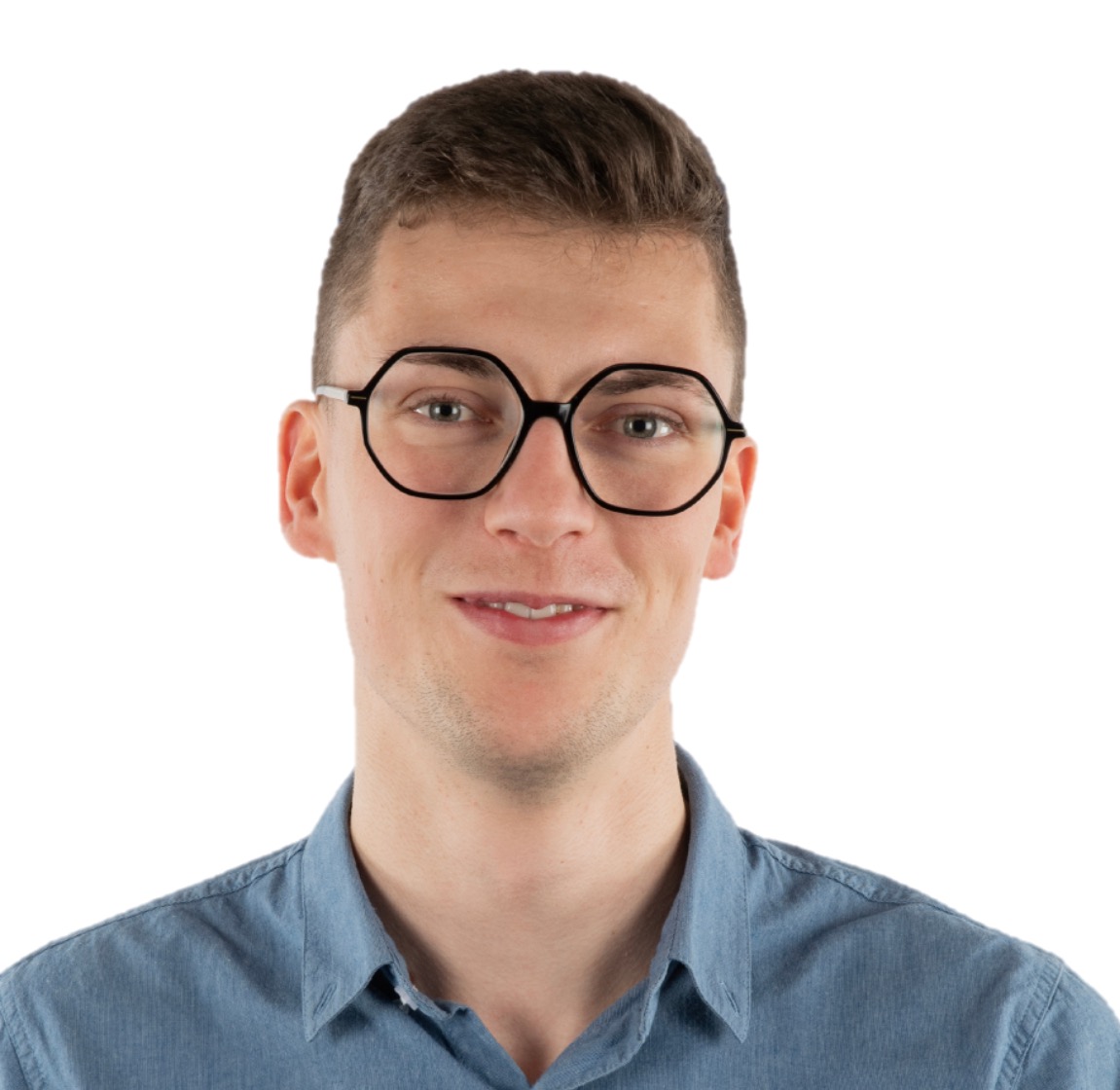}}]{Pietro Favaro}
received his Master’s degree in Electrical Engineering from the University of Mons, Belgium, in 2022, alongside a Master’s degree in Smart Cities and Communities from Heriot-Watt University, Scotland. He is currently pursuing a PhD at the University of Mons as an FRS-FNRS Fellow. In 2023, he was awarded a one-year fellowship from the Belgian-American Educational Foundation he spent at Johns Hopkins University. His research focuses on the scheduling of complex flexible assets, including pumped hydro energy storage and buildings.
\end{IEEEbiography}

\begin{IEEEbiography}[{\includegraphics[width=1in,height=1.25in,clip,keepaspectratio]{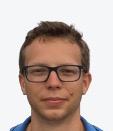}}]{Jean-François Toubeau}
(Member, IEEE) received the degree in civil electrical engineering and the Ph.D. degree in electrical engineering from the University of Mons, Belgium, in 2013 and 2018, respectively, where he is currently a Full-Time Senior Researcher within the “Power Systems and Markets Research Group.” His research mainly focuses on bridging the gap between machine learning and decision-making in modern power systems.
\end{IEEEbiography}


\begin{IEEEbiography}[{\includegraphics[width=1in,height=1.25in,clip,keepaspectratio]{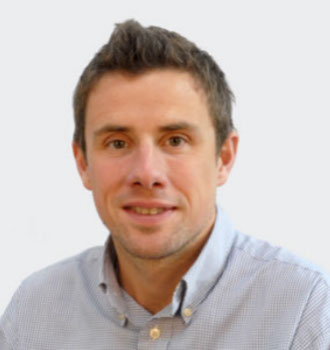}}]{François Vallée} received the degree in civil electrical engineering and the Ph.D. degree in electrical engineering from the Faculty of Engineering, University of Mons, Belgium, in 2003 and 2009, respectively, where he is currently a Full Professor and the Leader of the “Power Systems and Markets Research Group.” His research interests include PV and wind generation modeling for electrical system reliability studies in presence of dispersed generation, and adequacy studies. His Ph.D. work has been awarded the SRBE/KBVE Robert Sinave Award in 2010. He is currently serving as an Associate Editor for IEEE Transactions on Power Systems.
\end{IEEEbiography}

\begin{IEEEbiographynophoto}{Yury Dvorkin}
received his Ph.D. degree in Electrical Engineering from the University of Washington in 2016. He is currently an Associate Professor at Johns Hopkins University. His research focuses on modeling and algorithmic solutions for an efficient, reliable, and resilient energy transition.
\end{IEEEbiographynophoto}



\end{document}